\documentclass[12pt]{article}
\usepackage{hyperref} 
\hypersetup{colorlinks=true, allcolors=black} 
\usepackage{amsfonts,amsmath,amsxtra,amsthm}
\usepackage{amssymb}
\usepackage{bm}
\usepackage[utf8]{inputenc}
\usepackage{lscape}
\usepackage[normalem]{ulem}
\usepackage[russian,english]{babel}
\usepackage{xcolor}
\usepackage{pst-circ}
\usepackage{tikz, pgfplots}
\usepackage{tcolorbox}
\usepackage{mathrsfs}

\usepackage{thmtools,thm-restate}

\renewcommand{\epsilon}{\varepsilon}

\def\a{\alpha}

\def\C{\mathbb{C}}

\def\bg{\boldsymbol{\gamma}}

\def\beq{\begin{equation}}
\def\eeq{\end{equation}}
\def\beqq{\begin{equation*}}
\def\eeqq{\end{equation*}}

\def\bs{\begin{split}}
	\def\es{\end{split}}
\def\bg{\bm{\gamma}}
\def\bl{{\boldsymbol{\lambda}}}
\def\bbg{\underline{\bm{\gamma}}}

\def\bbl{\underline{\boldsymbol{\lambda}}}

\def\bo{\boldsymbol{\omega}}

\def\bx{\boldsymbol{x}}
\def\by{\boldsymbol{y}}

\def\bbx{\underline{\boldsymbol{x}}}
\def\bby{\underline{\boldsymbol{y}}}
\def\bbt{\underline{\boldsymbol{t}}}
\def\bbz{\underline{\boldsymbol{z}}}

\def\const{{2\pi\imath}}

\def\cbk{\color{black}}

\def\Im{\operatorname{Im}}
\def\g{\gamma}
\def\gg{{\hat{g}^\ast}}

\def\i{\hat{i}}

\def\K{{K}}
\def\KK{\hat{\K}}
\def\l{\lambda}

\def\mm{\hat{\mu}}

\def\o{\omega}

\def\R{\mathbb{R}}
\def\Re{\mathrm{Re}\,}

\def\ve{\varepsilon}
\def\vf{\varphi}

\newtheorem{lemma}{Lemma}
\newtheorem{proposition}{Proposition}
\newtheorem*{proposition*}{Proposition}
\newtheorem{corollary}{Corollary}
\newtheorem*{theorem*}{Theorem}

\newtheorem{remark}{Remark}
\newcommand{\rf}[1]{(\ref{#1})}

\newcommand{\LL}[1]{\Lambda_{#1}}
\def\LLL{\hat{\Lambda}}

\def\QQ{\hat{Q}}

\begin{document}
\begin{center}
{\bf \large Baxter operators in Ruijsenaars hyperbolic system II. \\[4pt] Bispectral wave functions}
\bigskip

{\bf N. Belousov$^{\dagger\times}$, S. Derkachov$^{\dagger\times}$, S. Kharchev$^{\bullet\ast}$, S. Khoroshkin$^{\circ\ast}$
}\medskip\\
$^\dagger${\it Steklov Mathematical Institute, Fontanka 27, St. Petersburg, 191023, Russia;}\smallskip\\
$^\times${\it National Research University Higher School of Economics, Soyuza Pechatnikov 16, \\St. Petersburg, 190121, Russia;}\smallskip\\
$^\bullet${\it National Research Center ``Kurchatov Institute'', 123182, Moscow, Russia;}\smallskip\\
$^\circ${\it National Research University Higher School of Economics, Myasnitskaya 20, \\Moscow, 101000, Russia;}\smallskip\\
$^\ast${\it Institute for Information Transmission Problems RAS (Kharkevich Institute), \\Bolshoy Karetny per. 19, Moscow, 127994, Russia}
\end{center}

\begin{abstract}
\noindent In the previous paper we introduced a commuting family of Baxter $Q$-operators for the quantum Ruijsenaars hyperbolic system. In the present work we show that the wave functions of the quantum system found by M. Hallnäs  and S. Ruijsenaars also diagonalize Baxter operators. Using this property we prove the conjectured duality relation for the wave function. As a corollary, we show that the wave function solves bispectral problems for pairs of dual Macdonald and Baxter operators. Besides, we prove the conjectured symmetry of the wave function with respect to spectral variables and obtain new integral representation for it.
\end{abstract}

\tableofcontents

\section{Introduction}
\subsection{Ruijsenaars system and Baxter $Q$-operators}
 In \cite{BDKK} we defined Baxter $Q$-operators for the Ruijsenaars hyperbolic system and proved their commutativity. Now we apply the results of \cite{BDKK} to the study of the wave functions of this system.

\color{black} Denote by $T^{a}_{x_i}$ the shift operator
\begin{equation}
T^{a}_{x_i}:=e^{a\partial_{x_i}}, \qquad \left( T^{a}_{x_i} \, f \right)(x_1,\ldots,x_i,\ldots,x_n) = f(x_1,\ldots,x_i+a,\ldots,x_n)
\end{equation}
and define its products for any subset $I\subset[n] = \{1, \dots, n\}$
\begin{equation}
T^{a}_{I,x}=\prod_{i \in I} T_{x_i}^a.
\end{equation}
The Ruijsenaars system \cite{R4} is governed by commuting symmetric difference operators
\begin{equation}
\label{I2}
H_r(\bx_n;g|\bo) = \sum_{\substack{I\subset[n] \\ |I|=r}}
\prod_{\substack{i\in I \\ j\notin I}}
\frac{\sh^{\frac{1}{2}}\frac{\pi}{\o_2}\left(x_i-x_j-\imath g\right)}
{\sh^{\frac{1}{2}}\frac{\pi}{\o_2}\left(x_i-x_j\right)}
\cdot T^{-\imath\o_1}_{I,x}\cdot \prod_{\substack{i\in I \\ j\notin I}}
\frac{\sh^{\frac{1}{2}}\frac{\pi}{\o_2}\left(x_i-x_j+\imath g\right)}
{\sh^{\frac{1}{2}}\frac{\pi}{\o_2}\left(x_i-x_j\right)}.
\end{equation}
Here and in what follows we denote tuples of $n$ variables as
 \begin{equation}
	\bm{x}_n = (x_1, \dots,x_n).
\end{equation}
One can also consider gauge equivalent Macdonald operators
\begin{equation}
	\label{I2a}
	M_r(\bx_n;g|\bo) = \sum_{\substack{I\subset[n] \\ |I|=r}}
	\prod_{\substack{i\in I \\ j\notin I}}
	\frac{\sh\frac{\pi}{\o_2}\left(x_i-x_j-\imath g\right)}
	{\sh\frac{\pi}{\o_2}\left(x_i-x_j\right)}
	\cdot T^{-\imath\o_1}_{I,x}.
\end{equation}
Both families of operators are parametrized by three constants: periods $\bm{\omega}=(\omega_1, \omega_2)$ and coupling constant $g$,
 which originally are supposed to be real positive. The equivalence is established by means of the measure function
\beq\label{I5}
\mu(\bx_n)=\prod_{\substack{i,j=1 \\ i\not=j}}^n\mu(x_i-x_j),\quad\text{where}\quad
\mu(x):=\mu_g(x|\bo)=S_2(\imath x|\bo) S_2^{-1}(\imath x+g|\bo). \eeq
Here $S_2(z|\bo)$ is the double sine function, see Appendix \ref{AppA}.
Namely,
\beq\label{I4}
\sqrt{\mu (\bx_n)} \,
M_r(\bx_n;g|\bo) \, \frac{1}{\sqrt{\mu  (\bx_n)}}=
H_r(\bx_n,g|\bo).
\eeq
Ruijsenaars operators are symmetric with respect to the sesquilinear scalar product defined by the Lebesque measure $d\bx_n$, while the Macdonald operators are symmetric with respect to the scalar product related to the measure
\beq\label{measure} \mu  (\bx_n) \, d\bx_n,\qquad \bm{x}_n\in\R^n.\eeq

In this paper, as well as in \cite{BDKK} and unlike the original Ruijsenaars setting, we consider periods $\bo$ and coupling constant $g$ to be complex valued,  assuming that
\beq\label{I0a} \Re \o_1 > 0, \qquad \Re \o_2 > 0,\qquad 0< \Re g<\Re \o_1+\Re \o_2  \eeq
and
\beq\label{I0b} \nu_g=\Re\frac{ g}{\o_1\o_2}>0.\eeq
Denote by $\K(x)$ the following function of a complex variable
\beq\label{I6} \K(x):=\K_g(x|\bo)=S_2^{-1}\Bigl(\imath x +\frac{g^\ast}{2}\Big|\bo\Bigr)S_2^{-1}\Bigl(-\imath x+\frac{g^\ast}{2}\Big|\bo\Bigr)\eeq
written in terms of the double sine fuction $S_2\bigl(z\big|\bo\bigr)$ and
\beq\label{I3b} g^\ast=\o_1+\o_2-g.\eeq
We also use the products of this function
\begin{equation}\label{I6a}
\K(\bx_n,\by_m)=\prod_{i=1}^n \prod_{j = 1}^m \K(x_i-y_j).
\end{equation}
In  \cite{BDKK} we introduced a family of operators $Q_{n}(\lambda)$ parameterized by $\lambda \in \mathbb{C}$ and called Baxter $Q$-operators. These are integral operators
\begin{equation}\label{I14}
\left( Q_{n}(\lambda) f\right) (\bm{x}_n) =  d_n(g|\bo) \cbk\, \int_{\mathbb{R}^n} d\bm{y}_n \, Q(\bm{x}_n, \bm{y}_n; \lambda) f(\bm{y}_n)
\end{equation}
with the kernel
\beq\label{Qker} Q(\bx_n,\by_{n};\l)= e^{\const \l(\bbx_n-\bby_n)}
\K(\bx_n,\by_{n})\mu  (\by_{n}), \qquad x_j,y_j \in \R
\eeq
and normalizing constant
\begin{equation}\label{dconst}
d_n(g|\bo) = \frac{1}{n!} \left[ \sqrt{\omega_1 \omega_2} S_2(g|\bo) \right]^{-n}.
\end{equation}
Here and below for a tuple $\bx_n=(x_1,\ldots,x_n)$ we use the notation $\bbx_n$ for the sum of components
\beq\label{I10a} \bbx_n=x_1+\ldots+x_n.
\eeq
We remark that the constant $d_n(g|\bo)$ was absent in our previous paper. We added it here to simplify formula for $Q$-operator's eigenvalue in the next subsection. The main result of \cite{BDKK} is given by the following theorem.
\begin{theorem*} \cite{BDKK}  Under assumptions \rf{I0a}, \rf{I0b} $Q$-operators commute
\begin{equation}\label{QQcomm}
Q_{n}(\lambda) \, Q_{n}(\rho) = Q_{n}(\rho) \, Q_{n}(\lambda).
\end{equation}
The kernels of the operators in both sides of \rf{QQcomm} are analytic functions of $\l,\rho$ in the strip
\beq \label{QQcond}  |\Im(\lambda - \rho)| < \nu_g. \eeq
\end{theorem*} 
\begin{remark} The product of two $Q$-operators is a well-defined integral operator on the domain of fast decreasing functions, see Proposition 5 and remark after it in \cite{BDKK}. This is also true for the various products of $Q$- and $\Lambda$-operators below.
\end{remark} 

We also proved that these operators commute with Macdonald operators
\begin{equation}\label{QMcomm}
Q_{n}(\lambda) \, M_r(\bm{x}_n; g|\bo) = M_r(\bm{x}_n; g|\bo) \, Q_{n}(\lambda), \qquad r=1, \ldots, n
\end{equation}
under condition
\begin{equation}\label{I13a}
\Re g < \Re \o_2.
\end{equation}
This suggests that they should have common eigenfunctions.

\subsection{Wave functions and their symmetries}

Denote by $\LL{n}(\l)$ the integral operator similar to the Baxter $Q$-operator
\beq\label{I11}\begin{split} \left(\LL{n}(\l)f\right)(\bx_n)=d_{n - 1}(g|\bo) \int_{\R^{n-1}}d\by_{n-1} \, \Lambda(\bx_n,\by_{n-1};\l) f(\by_{n-1})
\end{split}\eeq
with the kernel
\beq\label{Lker} \Lambda(\bx_n,\by_{n-1};\l)= e^{\const \l(\bbx_n-\bby_{n-1})}
\K(\bx_n,\by_{n-1})\mu  (\by_{n-1}), \qquad x_j,y_j \in \R
\eeq
and constant $d_{n - 1}(g|\bo)$ given by the formula \eqref{dconst}.

M. Hallnäs  and S. Ruijsenaars proved  in \cite{HR1} that for real positive periods $\bo$ and the coupling constant
$g$,  $0<\Re g<\o_2$,  the function
\beq \label{I12}\Psi_{\bl_n}(\bx_n):=\Psi_{\bl_n}(\bx_n; g|\bo)=\LL{n}(\l_n) \, \LL{n-1}(\l_{n-1})\cdots \LL{2}(\l_2) \, e^{\const \l_1x_1}
\eeq
is a joint eigenfunction of Macdonald operators. Similar arguments show that the same statement holds for complex valued periods $\bo$ and the coupling constant $g$.
\begin{restatable}{thm}{theor}
	\label{th0} Under the conditions \rf{I0a},  \rf{I0b} and \rf{I13a} we have the equality
\beq \label{I13} M_r(\bx_n;g|\bo)\,\Psi_{\bl_n}(\bx_n)= e_r \bigl(e^{{2\pi \lambda_1}{\o_1}}, \dots, e^{{2\pi \lambda_n}{\o_1}}\bigr) \cbk \Psi_{\bl_n}(\bx_n).
\eeq
\end{restatable}
Here $e_r(z_1,\ldots,z_n)$ is $r$-th elementary symmetric function,
\beqq e_r(z_1,\ldots,z_n)=\sum_{1\leq i_1<i_2<\ldots< i_r\leq n}z_{i_1}\cdots z_{i_r}, \eeqq

 The analytical properties of the function $\Psi_{\bl_n}(\bx_n)$ with respect to $x_j, \lambda_j$ (in the case of positive periods) are investigated in detail in \cite{HR1}. Here we only remark that from definition \eqref{I12} and formula for the kernel \eqref{Lker} the variables $x_j$ determine the positions of poles of functions $K(x_j-y_i)$, see \eqref{S-zeros}. In the case $x_j \in \R$ the sequences of poles are separated by the integration contours. The closest poles to the contours are $y_i = x_j \pm \imath g^*/2$.

Therefore, the wave function $\Psi_{\bl_n}(\bx_n)$ extends to the analytic function of $x_j \in \mathbb{C}$ in the domain
\begin{equation}
| \Im x_j | < \frac{g^*}{2}, \qquad j=1, \dots,n.
\end{equation}
From the convergence of integral representation \eqref{I12} we also have the  analyticity of the wave function with respect to the variables $\lambda_j \in \mathbb{C}$ in the domain
\begin{equation}
| \Im(\lambda_j - \lambda_k) | <\theta, \qquad j,k =1, \dots,n
\end{equation}
with some positive $\theta$, see Proposition \ref{Prop-Qconv} in Section \ref{sec-Qeigen}.

The short proof of Theorem \ref{th0}, combining the arguments by M. Hallnäs and S. Ruijsenaars \cite{HR1} and by I. Macdonald \cite{M} is given in Section \ref{sec-Macdonald}.  See also \cite{KK2} for an analogous proof in the case of hyperbolic Calogero-Sutherland model, which represents a non-relativistic limit of the Ruijsenaars system. \cbk

\begin{remark}
Note that the double sine function is invariant under permutation of $\omega_1, \omega_2$, therefore the wave function $\Psi_{\bm{\lambda}_n}(\bx_n)$ is invariant too. This implies that under condition
\begin{equation}
0 < \Re g < \min(\Re \omega_1, \Re \omega_2)
\end{equation}
it is also an eigenfunction of difference Macdonald operators $M_r(\bm{x}_n; g | \omega_2, \omega_1)$ related to the shift of variables by $\imath \omega_2$, and thus satisfies two pairs of compatible difference equations.
\end{remark}

The formula \eqref{I12} also can be written recursively
\begin{equation}\label{Psi-rec}
\Psi_{\bm{\lambda}_n}(\bm{x}_n) = \Lambda_{n}(\lambda_n) \, \Psi_{\bm{\lambda}_{n - 1}}(\bm{x}_{n - 1}), \qquad \Psi_{\lambda_1}(x_1) = e^{{2\pi \imath}{} \lambda_1 x_1}.
\end{equation}
Hence, we call $\Lambda_{n}(\lambda)$ raising operator.

 Our approach to study the wave functions in Ruijseenars system is based on the detailed investigation of the relations between the various products of $Q$-operators \rf{I14} and raising operators \rf{I11} which follow from commutativity relation of $Q$-operators~\rf{QQcomm}.

Firstly, an immediate consequence of \rf{QQcomm} is the exchange relation for the raising operators, its short proof is given in Section \ref{sec-QL}.
\begin{restatable}{thm}{theoremll}\label{theoremll}
Under assumptions \rf{I0a}, \rf{I0b} raising operators satisfy relation
\begin{equation}\label{ll}
	\Lambda_{n}(\lambda) \, \Lambda_{n - 1}(\rho) = \Lambda_{n }(\rho) \, \Lambda_{n - 1}(\lambda).
\end{equation}
The kernels of the operators in both sides of \rf{ll} are analytic functions of $\l,\rho$ in the strip
\beq |\Im(\lambda - \rho)| < \nu_g. \eeq
\end{restatable}
\begin{corollary}\label{corsym}
	The wave function $\Psi_{\bm{\lambda}_n}(\bm{x}_n)$ is invariant under permutations of $\bm{x}_n$ and $\bm{\lambda}_n$ components.
\end{corollary}
\begin{proof}
	The symmetry with respect to permutations of $x_j$ follows immediately from the defining formula \eqref{Psi-rec}. The symmetry over spectral variables $\l_j$ is a consequence of Theorem \ref{theoremll}. To apply it we use absolute convergence of the corresponding multiple integral~\eqref{I12} (see Proposition \ref{Prop-Qconv} in Section \ref{sec-Qeigen}) that allows us to change integration order in this integral.
\end{proof}
The next exchange relation appears in the limit of the $Q$-operators commutation relation \eqref{QQcomm} as one of the external kernel parameters $z_n$ tends to infinity. Set
\beq\label{I27a} \hat{a}=\frac{a}{\o_1\o_2}\eeq
for any $a\in\C$, so that
\beq \hat{\bo}=\left(\frac{1}{\o_2},\frac{1}{\o_1}\right) ,
\qquad \hat{g}=\frac{g}{\o_1\o_2},\qquad  \gg=\hat{\o}_1 + \hat{\o}_2 - \hat{g} = \frac{g^\ast}{\o_1\o_2}.
\eeq
Also introduce the function
\beq\label{I6c} \KK(\l):=\K_{\gg}(\l|\hat{\bo})=S_2^{-1}\Bigl(\imath \l +\frac{\hat{g}}{2}\Big|\hat{\bo}\Bigr)S_2^{-1}\Bigl(-\imath \l+\frac{\hat{g}}{2}\Big|\hat{\bo}\Bigr).\eeq
\begin{restatable}{thm}{QLambda}\label{theoremQL}
	The operator identity
	\begin{equation}\label{QLcomm}
		Q_{n}(\lambda) \, \Lambda_{n}(\rho) = \KK(\lambda - \rho) \; \Lambda_{n}(\rho) \, Q_{n-1}(\lambda)
	\end{equation}
	holds true for $\lambda, \rho \in \mathbb{C}$ such that
	\beq \label{QLcond} | \Im(\lambda - \rho) | < \frac{\nu_g}{2}.\eeq
\end{restatable}
The  proof is given in Section \ref{sec-QL}.  Note that the strip for $\lambda,\rho$ \eqref{QLcond} is twice narrower than in the $Q$-commutativity relation \eqref{QQcond}. Using the iterative representation of the wave function \eqref{I12} and the relation \rf{QLcomm} we arrive at the spectral description of the Baxter $Q$-operator. Its proof with the necessary convergence arguments is given in Section~\ref{sec-Qeigen}.
\begin{restatable}{thm}{QPsi}\label{thm-QPsi}
	The wave function $\Psi_{ \bm{\lambda}_n }(\bm{x}_n)$ is a joint eigenfunction of the commuting family of operators $Q_{n}(\lambda)$
	\begin{equation}\label{QPsi}
		Q_{n}(\lambda) \, \Psi_{ \bm{\lambda}_n }(\bm{x}_n) = \prod_{j = 1}^n \KK(\lambda-\lambda_j) \, \Psi_{ \bm{\lambda}_n }(\bm{x}_n).
	\end{equation}
	The integrals in both sides of \rf{QPsi} converge if
	\begin{equation}
		| \Im(\lambda - \lambda_n) | < \frac{\nu_g}{2}(1-\ve), \qquad | \Im (\lambda_k - \lambda_j) | <\theta(\ve), \qquad k, j = 1, \dots, n
	\end{equation}
	for any $\ve \in[0,1)$ and
	\beq \theta(\ve) = \frac{\nu_g }{4 (n - 1)! e}\ve. \eeq
\end{restatable}

The kernels \rf{I14} and \rf{Lker} of the $Q$-operator and raising $\Lambda$-operator are related in two ways. Firstly, the kernel  $\Lambda(\bx_n,\by_{n-1};\l)$ appears in the certain limit of the kernel $Q(\bx_n,\by_{n};\l)$. Besides, the kernel $\Lambda(\bx_n,\by_{n-1};\l)$ equals the kernel $Q(\bx_{n-1},\by_{n-1};\l)$ multiplied by the function
\begin{equation}
e^{\const \l x_n} \prod_{j=1}^{n - 1}K(x_n - y_j)
\end{equation}
regular on the integration cycle. Then the  inductive use of Theorem \ref{thm-QPsi} allows to relate two integral representations for the wave function: the representation iterative with respect to variables $x_j$ and another one iterative with respect to spectral variables $\l_j$. This can be formulated as duality relation of the wave function conjectured by M. Hallnäs and S.~Ruijsenaars in \cite{HR1}.

For this purpose we introduce the counterparts of basic functions \rf{I5}, \rf{I6a} and operators \rf{I14}, \rf{I11}, but in the space of spectral variables $\l_j$. Namely, similarly to the function \eqref{I6c} we define the function
 \beq\label{I5d}
 \mm(\l):=\mu_{\hat{g}^*}(\l|\hat{\bo})=S_2(\imath \l|\hat{\bo}) S_2^{-1}(\imath \l+\gg|\hat{\bo}) \eeq
and their products
\begin{equation}\label{I6d}
\mm(\bl_n)=\prod_{\substack{i,j = 1\\i\not=j}}^n \mm(\l_i-\l_j),\qquad	\KK(\bl_n,\bg_m)=\prod_{i=1}^n \prod_{j = 1}^m \KK(\l_i-\g_j).
\end{equation}
Next we introduce counterparts of the $Q$-operator and raising $\Lambda$-operator that act on functions of spectral variables $\lambda_j$
\begin{equation}\label{I15d}\begin{aligned}
	\bigl( \QQ_n(x) f\bigr) (\bm{\l}_n)& = d_n(\hat{g}^*|\hat{\bo})\, \int_{\mathbb{R}^n} d\bm{\gamma}_n \, \QQ(\bm{\l}_n, \bm{\gamma}_n; x) f(\bm{\gamma}_n),\\[7pt]
	 \bigl(\LLL_{n}(x)f\bigr)(\bl_n)&=d_{n - 1}(\hat{g}^*|\hat{\bo}) \int_{\R^{n-1}}d\bg_{n-1} \, \LLL(\bl_n,\bg_{n-1};x) f(\bg_{n-1})
	\end{aligned}
\end{equation}
with the kernels
\beq\label{I14d}\begin{split} \QQ(\bl_n,\bg_{n};x)&= e^{\const x (\bbl_n-\bbg_n)}	\KK(\bl_n,\bg_{n})\hat{\mu}(\bg_{n}),\\[7pt]
\LLL(\bl_n,\bg_{n-1};x)&= e^{\const x (\bbl_n-\bbg_{n-1} )} \KK(\bl_n,\bg_{n-1}) \hat{\mu}  (\bg_{n-1}).
\end{split}
\eeq
 Additionally, in what follows we assume the condition
\beq \nu_{g^*}:= \Re \hat{g}^* = \Re \frac{g^*}{\o_1 \o_2} > 0.\eeq
Then together with the conditions \eqref{I0a}, \eqref{I0b} we have two pairs of double inequalities for the coupling constant
\begin{equation}\label{gg-ineq}
0<\Re g < \Re\o_1 + \Re\o_2, \qquad 0< \Re \hat{g}^* < \Re\hat{\o}_1 + \Re\hat{\o}_2.
\end{equation}
Note that for the real periods they are equivalent.  Clearly, with these conditions introduced operators enjoy relations analogous to \rf{ll}, \eqref{QLcomm} and can be used for iterative construction of the function
\begin{equation}\label{Psi-rec2}
	\hat{\Psi}_{\bx_n}(\bl_n) := \Psi_{\bx_n}(\bl_n;\gg|\hat{\bo}) = \LLL_{n}(x_n) \, \hat{\Psi}_{\bx_{n - 1}}(\bl_{n - 1}), \qquad \hat{\Psi}_{x_1}(\l_1) = e^{{2\pi \imath}{} \lambda_1 x_1}.
\end{equation}
The following statement was conjectured in \cite{HR1}.
\begin{restatable}{thm}{duality}\label{thm-duality}
	The wave function $\Psi_{\bm{\lambda}_n}(\bm{x}_n)$ satisfies duality relation
	\begin{equation}\label{duality}
		\Psi_{\bm{\lambda}_n}(\bm{x}_n; g|\bo) = \Psi_{\bm{x}_n}(\bm{\lambda}_n; \gg|\hat{\bo}).
	\end{equation}
\end{restatable}
The proof of this theorem is given in Section \ref{sec-duality}. It has several direct corollaries.
\begin{corollary} \label{corQQ} The wave function  $	\Psi_{\bm{\lambda}_n}(\bm{x}_n)$ is an eigenfunction of both operators $Q_{n}(\lambda)$ and $\QQ_{n}(x)$
\begin{equation}\label{QQPsi}
	\begin{aligned}
	Q_{n}(\lambda) \, \Psi_{ \bm{\lambda}_n }(\bm{x}_n) &= \prod_{j = 1}^n \KK(\lambda-\lambda_j) \, \Psi_{ \bm{\lambda}_n }(\bm{x}_n),\\
	\QQ_{n}(x) \, \Psi_{ \bm{\lambda}_n }(\bm{x}_n) &= \prod_{j = 1}^n \K(x-x_j) \, \Psi_{ \bm{\lambda}_n }(\bm{x}_n).
	\end{aligned}
\end{equation}
\end{corollary}
\begin{corollary}\label{bispectral} Under the condition $\Re g < \Re\o_2$ \rf{I13a} and analogous one $\Re \hat{g}^* < \Re \hat{\o}_2$  the wave function  $\Psi_{\bm{\lambda}_n}(\bm{x}_n)$ is an eigenfunction of both families of Macdonald operators $M_r(\bx_n;g|\bo)$ and $M_s(\bl_n;\hat{g}^*|\hat{\bo})$
	\beq\label{MM}\begin{split}
	 M_r(\bx_n;g|\bo)\,\Psi_{\bl_n}(\bx_n)&=e_r \bigl(e^{{2\pi \lambda_1\o_1}}, \dots, e^{{2\pi \lambda_n\o_1}} \bigr)\Psi_{\bl_n}(\bx_n),\\[5pt]
	 M_s(\bl_n;\hat{g}^*|\hat{\bo})\, \Psi_{\bl_n}(\bx_n)&=e_s \bigl(e^{\frac{2\pi x_1}{\o_2}}, \dots, e^{\frac{2\pi x_n}{\o_2}} \bigr)\Psi_{\bl_n}(\bx_n).
\end{split}
\eeq
\end{corollary}

In other words, the function $\Psi_{\bm{\lambda}_n}(\bm{x}_n)$ solves bispectral problems for pairs of dual \break $Q$-operators and for a pair of dual Macdonald operators.  In non-relativistic limit the bispectrality property of the wave function of Sutherland model was established in \cite{KK1,KK2}. The  duality \rf{MM} \cbk provides us with two integral representations \eqref{Psi-rec}, \eqref{Psi-rec2}, recursive with respect to $x_j$ and $\lambda_j$ variables correspondingly.
\begin{equation}
\Psi_{\bm{\lambda}_n}(\bm{x}_n) = \Lambda_{n}(\lambda_n) \, \Psi_{\bm{\lambda}_{n - 1}}(\bm{x}_{n - 1}) = \LLL_{n}(x_n) \, \Psi_{\bm{\lambda}_{n - 1}}(\bm{x}_{n - 1}).
\end{equation}
For real periods these two recursive procedures can be combined into the third one, symmetric with respect to coordinates and spectral parameters, see Section \ref{sec-duality} for details.
\setcounter{crl}{3}
\begin{restatable}{crl}{PsiQQ}
 	 For real positive periods $\o_1, \o_2$ the wave function admits the iterative integral representation
 	\begin{equation}
 		\begin{aligned}
 			\Psi_{\bm{\lambda}_n}(\bm{x}_n) &= e^{\const \lambda_n x_n} Q_{n - 1}(\lambda_n) \, \QQ_{n - 1}(x_n) \, \Psi_{\bm{\lambda}_{n - 1}}(\bm{x}_{n - 1} ) \\[6pt]
 			&= Q_{n - 1}(\lambda_n) \, \QQ_{n - 1}(x_n) \cdots Q_{1}(\lambda_2) \, \QQ_{1}(x_2) \, e^{\const (\lambda_1 x_1 + \ldots + \lambda_n x_n)}.
 		\end{aligned}
 	\end{equation}
\end{restatable}
An analogous representation is presented in \cite{GLO} for the open Toda wave function. The wave function of open Toda chain also admits two integral representations iterative with respect to coordinates $x_j$ and spectral variables $\l_j$ correspondingly. In \cite{GLO} their equivalence is established in a similar way with the help of two dual $Q$-operators and their spectral descriptions analogous to \eqref{QQPsi}. But these spectral properties are proved differently.

Finally, we note that in the works of M. Hallnäs and S. Ruijsenaars \cite{HR1}, \cite{HR2}, \cite{HR3}, where only the case of real periods $\o_1, \o_2$ is considered, the wave function $\Phi_{\bl_n}(\bx_n)$ differs from ours by renormalization of spectral parameters, that is
\begin{equation}
\Phi_{\bl_n}(\bx_n) = \Psi_{ \bl_n / \o_1 \o_2 }(\bx_n).
\end{equation}
Therefore, due to homogeneity of the double sine function \eqref{S-hom} (with $\gamma = \o_1 \o_2$), the duality relation \eqref{duality} for this function looks as
\begin{equation}\label{HR-duality}
\Phi_{\bl_n}(\bx_n; g | \bo) = \Phi_{\bx_n}(\bl_n; g^*|\bo ).
\end{equation}
However, for the complex periods the scaling by $\o_1 \o_2$ rotates integration contours in the dual integral representation. Note that for $n=2,3$ the relation \eqref{HR-duality}  was proved  by different method in \cite{HR2} .

\setcounter{equation}{0}
\section{Exchange relations}\label{sec-QL}
In this section we derive exchange relations between $Q$-operators and raising operators,
as well as between raising operators themselves, starting from basic commutativity property of $Q$-operators. We start with exchange relation between raising operators.
\theoremll*
\begin{proof}
	In terms of kernels \eqref{Lker} the identity \rf{ll} is equivalent to the relation
	\begin{equation}
		\Lambda(\bm{x}_n, \bm{z}_{n - 2}; \lambda, \rho) = e^{\const (\rho - \lambda) (\bbx_n + \bbz_{n - 2} ) } \, \Lambda(\bm{x}_n, \bm{z}_{n - 2}; \rho, \lambda),
	\end{equation}
	where
	\begin{equation}
		\Lambda(\bm{x}_n, \bm{z}_{n - 2}; \lambda, \rho) = \int_{\mathbb{R}^{n - 1}} d\bm{y}_{n - 1} \, \mu  (\bm{y}_{n - 1}) \, e^{ \const (\rho - \lambda) \bby_{n-1} } \, \K(\bm{x}_n, \bm{y}_{n - 1}) \, \K(\bm{y}_{n - 1}, \bm{z}_{n - 2}).
	\end{equation}
	Denoting $x_n := z_{n - 1}$ we arrive at $Q$-operators commutation relation \eqref{QQcomm}
	\begin{equation}
		Q_{n - 1}(\lambda) \, Q_{n - 1}(\rho) = Q_{n - 1}(\rho) \, Q_{n - 1}(\lambda)
	\end{equation}
	written in terms of kernels.
\end{proof}

The following theorem plays the crucial role in further derivations of the duality properties of the wave functions studied in this paper.
\QLambda*		
\begin{proof}
Recall $Q$-operators commutativity relation \eqref{QQcomm}
\begin{equation}\label{QQcomm2}
Q_{n}(\lambda) \, Q_n(\rho) = Q_n(\rho) \, Q_n(\lambda).
\end{equation}
The main idea of the current proof is to take the limit of the identity \eqref{QQcomm2} written in terms of kernels
\begin{equation}\label{QQker}
Q(\bm{x}_n, \bm{z}_n; \lambda, \rho) = Q_{}(\bm{x}_n, \bm{z}_n; \rho, \lambda)
\end{equation}
as $z_n \rightarrow \infty$, where
\begin{equation}
Q(\bm{x}_n, \bm{z}_n; \lambda, \rho) = \int_{\mathbb{R}^n} d\bm{y}_n \, Q(\bm{x}_n, \bm{y}_n; \lambda) \,Q(\bm{y}_n, \bm{z}_n; \rho), \qquad x_j, z_j \in \mathbb{R}.
\end{equation}
Here we canceled constants $d_n^2(g|\bo)$ from both sides. However, to obtain an appropriate limit we should first multiply \eqref{QQker} by the function
\begin{equation}
r(\bm{z}_n; \rho) = \exp\Bigl( \pi \hat{g} \bigl[\bbz_{n - 1}+(2-n)z_n\bigr] + {2 \pi}\imath \rho  z_n \Bigr),
\end{equation}
where we recall the notation \eqref{I27a}
$$\hat{g} = \frac{g}{\o_1 \o_2}.$$
We note that the integral \eqref{QQker} is absolutely convergent for $\lambda, \rho \in \mathbb{C}$ such that
\beq \label{IV8a} | \Im(\lambda - \rho) | < \nu_g = \Re \hat{g}, \eeq
see \cite[Proposition 5]{BDKK}. So, in what follows we assume this condition for $\l,\rho$.

Consider the left-hand side of \eqref{QQker}. Multiplying by the function $r(\bm{z}_n; \rho)$, we can rewrite it as
\begin{equation}\label{IV13}
r(\bm{z}_n; \rho) \, Q_{}(\bm{x}_n, \bm{z}_n; \lambda, \rho) = \int_{\mathbb{R}^n} d\bm{y}_n F(\bm{x}_n, \bm{y}_n, \bm{z}_n; \lambda, \rho),
\end{equation}
where
\begin{equation}\label{F}
\begin{aligned}
F &= e^{{2 \pi \imath}  \lambda (\bbx_n - \bby_n) + 2\pi\imath\left( \rho - \frac{\imath \hat{g}}{2} \right) ( \bby_n - \bbz_{n - 1} )  } \\[10pt]
&\times \K (\bm{x}_n, \bm{y}_n) \, \mu  (\bm{y}_n) \, \K (\bm{y}_n, \bm{z}_{n - 1}) \, \mu  (\bm{z}_{n - 1}) \\[5pt]
&\times \prod_{j = 1}^n e^{ \pi\hat{g} (z_n - y_j) } \K (y_j - z_n) \; \prod_{j = 1}^{n - 1} e^{ 2\pi \hat{g} (z_j - z_n) } \mu  (z_n - z_j) \mu(z_j - z_n).
\end{aligned}
\end{equation}
The variable $z_n$ is contained only in the last line. The function $r(\bm{z}_n; \rho)$ was chosen such that products in the last line have a pointwise limit
\begin{equation}
\lim_{z_n \rightarrow \infty} \, \prod_{j = 1}^n e^{ \pi\hat{g} (z_n - y_j) } \K (y_j - z_n) \; \prod_{j = 1}^{n - 1} e^{ 2\pi \hat{g} (z_j - z_n) } \mu  (z_n - z_j) \mu(z_j - z_n) = 1.
\end{equation}
To evaluate this limit we use asymptotic formulas \eqref{Kmu-asymp}
\begin{equation}\label{IV8}
\mu(x) \sim e^{\pi \hat{g} | x | \pm \imath \frac{\pi \hat{g} g^*}{2} }, \qquad K(x) \sim e^{- \pi \hat{g} |x|}, \qquad x\rightarrow \pm \infty.
\end{equation}
And for sufficiently large $z_n $, such that $z_n > z_j$ for all $j = 1, \dots, n - 1$, it is bounded from above
\begin{equation}\label{IV15}\left|
\prod_{j = 1}^n e^{ \pi\hat{g} (z_n - y_j) } \K (y_j - z_n) \; \prod_{j = 1}^{n - 1} e^{ 2\pi \hat{g} (z_j - z_n) } \mu  (z_n - z_j) \mu(z_j - z_n) \right| \leq C_1
\end{equation}
with some $C_1$ which depends on $g$ and $\bo$. For this we use bounds \eqref{Kmu-bound}
\begin{equation}\label{IV16}
| \mu  (x) | \leq C e^{\pi\nu_g |x| }, \qquad |\K (x)| \leq C e^{- \pi\nu_g |x| }, \qquad x \in \mathbb{R}
\end{equation}
where again we denoted
\beq \nu_g=\Re \hat{g} >0 \eeq
and $C$ is a constant which depends on $g,\bo$.

The next step is to derive a bound for the whole function $F$ in order to use dominated convergence theorem and evaluate $z_n \rightarrow \infty$ limit of the integral \eqref{IV13}. For all functions in the second line \eqref{F} use bounds \eqref{IV16} together with triangle inequalities
\begin{equation}\label{Kmu-bounds2}
\begin{aligned}
|\K (x_i - y_j)| &\leq C \, e^{-\pi\nu_g |x_i - y_j|} \leq C \, e^{-\pi\nu_g (|y_j| - |x_i|)}, \\[6pt]
|\mu  (y_i - y_j)| &\leq C \, e^{\pi\nu_g |y_i - y_j|} \leq C \, e^{\pi\nu_g(|y_i| + |y_j|)}.
\end{aligned}
\end{equation}
Hence, using also bound \eqref{IV15} for the whole function $F$ we obtain
\begin{equation}
| F | \leq C_2(g, \bm{\omega}, \bm{x}_n, \bm{z}_{n - 1}) \; \exp \Bigl( \bigl[ \, | 2 \pi \Im(\lambda - \rho) + \pi\nu_g | - \pi\nu_g \bigr] \| \bm{y}_n \| \Bigr)
\end{equation}
with some $C_2$, where by $\|\bm{y}_n\|$ we mean $L^1$-norm
\begin{equation}
\| \bm{y}_n \| = \sum_{j = 1}^n |y_j|.
\end{equation}
Function from the right is integrable when
\beq \label{IV19} \Im(\lambda - \rho) \in (-\nu_g, 0). \eeq
Note that this condition is stronger than the one assumed before \eqref{IV8a}. So, assuming it we use dominated convergence theorem to evaluate the limit of the integral \eqref{IV13}
\begin{equation}
\begin{aligned}
\lim_{z_n \rightarrow \infty} \, r(\bm{z}_n; \rho) \, Q_{}(\bm{x}_n, \bm{z}_n; \lambda, \rho) &= \int_{\mathbb{R}^n} d\bm{y}_n \; e^{2 \pi \imath \left[ \lambda (\bbx_n - \bby_n) + \left( \rho - \frac{\imath \hat{g}}{2} \right) ( \bby_n - \bbz_{n - 1} ) \right] } \\[10pt]
&\times \K (\bm{x}_n, \bm{y}_n) \, \mu  (\bm{y}_n) \, \K (\bm{y}_n, \bm{z}_{n - 1}) \, \mu  (\bm{z}_{n - 1}).
\end{aligned}
\end{equation}
The result coincides with the kernel of operator $Q_n(\lambda) \, \Lambda_{n}(\rho -  \imath \hat{g}/2)$ modulo multiplication by constant~$d_n(g|\bo) d_{n - 1}(g|\bo)$.

The next step is to take the limit of the right-hand side of equality \eqref{QQker} multiplied by the function $r(\bm{z}_n; \rho)$
\begin{equation}
r(\bm{z}_n; \rho) \, Q_{}(\bm{x}_n, \bm{z}_n; \rho, \lambda) = \int_{\mathbb{R}^n} d\bm{y}_n \; G(\bm{x}_n, \bm{y}_n, \bm{z}_n; \rho, \lambda).
\end{equation}
In contrast with the function $F$ appeared on the left-hand side, the function $G$ doesn't have a pointwise limit as $z_n \rightarrow \infty$. The trick is to perform a certain shift of integration variable $y_n$ and then take the limit.

But before we also need to slightly rewrite this integral. Notice that the function $G$ is symmetric with respect to variables $y_j$. Define the domain $D_j$ as
\begin{equation}
D_j = \{ \bm{y}_n \in \mathbb{R}^n \colon y_j \geq y_k, \; \forall k \in [n] \setminus \{j\} \}.
\end{equation}
We have the equality
\begin{equation}
\bm{1}_{D_1} + \bm{1}_{D_2} + \ldots + \bm{1}_{D_n} = 1,
\end{equation}
where $\bm{1}_{D_j}$ is the indicator function of the domain $D_j$. Therefore using symmetry of $G$ we can rewrite the last integral
\begin{equation}
r(\bm{z}_n; \rho) \, Q_{}(\bm{x}_n, \bm{z}_n; \rho, \lambda) = n \int_{\mathbb{R}^n} d\bm{y}_n \; \bm{1}_{D_n} \, G(\bm{x}_n, \bm{y}_n, \bm{z}_n; \rho, \lambda),
\end{equation}
where we gained factor of $n$ due to the equality above.

Now we should shift the integration variable $y_n \rightarrow y_n + z_n$. Let us write the result after the shift
\begin{equation}\label{IV24}
\begin{aligned}
r(& \bm{z}_n; \rho) \, Q_{}(\bm{x}_n, \bm{z}_n; \rho, \lambda) = n \int_{\mathbb{R}^n} d\bm{y}_n \; e^{2 \pi \imath \left[ \left( \rho - \frac{\imath \hat{g}}{2} \right) (\bbx_n - \bby_n) + \lambda ( \bby_n - \bbz_{n - 1} ) \right] } \\[10pt]
&\times \K (\bm{x}_n, \bm{y}_{n - 1}) \, \mu  (\bm{y}_{n - 1}) \, \K (\bm{y}_{n - 1}, \bm{z}_{n - 1}) \, \K (y_n) \, \mu  (\bm{z}_{n - 1}) \, R(\bm{x}_n, \bm{y}_n, \bm{z}_n),
\end{aligned}
\end{equation}
where by $R$ we denoted all factors that contain $z_n$
\begin{equation}
\begin{aligned}
R&(\bm{x}_n, \bm{y}_n, \bm{z}_n) =  \prod_{j = 1}^n e^{ \pi \hat{g} (y_n + z_n - x_j) } \K (x_j - y_n - z_n) \prod_{j = 1}^{n - 1} e^{ 2\pi \hat{g} (y_{jn} - z_n) } \mu   (y_{nj} + z_n) \mu(y_{jn} - z_n) \\[3pt]
&\times\prod_{j = 1}^{n - 1} e^{ \pi\hat{g}(z_{nj} + y_{nj} + z_n) } \K (y_n + z_{nj}) \, \K (y_j - z_n) \; \prod_{j = 1}^{n - 1} e^{ 2 \pi \hat{g}z_{jn}}  \mu  (z_n - z_j) \mu(z_j - z_n) \cdot \bm{1}_{D_n'}
\end{aligned}
\end{equation}
Here, for brevity, we introduced notation $y_{nj} = y_n - y_j$. Domain of the indicator function has changed after the shift
\begin{equation}
D'_n = \{ \bm{y}_n \in \mathbb{R}^n \colon y_n + z_n \geq y_k, \; \forall k \in [n-1] \}.
\end{equation}
Due to the asymptotics \eqref{IV8} it has a pointwise limit
\begin{equation}
\lim_{z_n \rightarrow \infty} R(\bm{x}_n, \bm{y}_n, \bm{z}_n) = 1.
\end{equation}
Hence, to use the dominated convergence theorem we need to write a bound independent of $z_n$ for it. Note that in the presence of the indicator function we have equality
\begin{equation}
y_{nj} + z_n = |y_{nj} + z_n|, \qquad j = 1, \ldots, n - 1.
\end{equation}
Therefore, for a factor from the second product in $R$ we can write
\begin{equation}
\left|e^{ 2\pi \hat{g} (y_{jn} - z_n) } \mu   (y_{nj} + z_n) \mu(y_{jn} - z_n)\right| \leq C^2,
\end{equation}
where we used the bound for $\mu $ \eqref{IV16}. Factors from three other products can be estimated similarly using bounds on $\K $ and $\mu$ \eqref{IV16} and assuming  that $z_n$ is sufficiently large, such that $z_n > z_j$ for all $j = 1, \dots, n - 1$. We conclude that $R$ is bounded
\begin{equation}
|R(\bm{x}_n, \bm{y}_n, \bm{z}_n)| \leq C_3(g, \bm{\omega})
\end{equation}
with some $C_3$.

Now using this fact together with bounds \eqref{Kmu-bounds2} we estimate the integrand in \eqref{IV24}
\begin{equation}
\begin{aligned}
\bigl| \bm{1}_{D'_n} \, &G(\bm{x}_n, \bm{y}_{n - 1}, y_n + z_n, \bm{z}_n; \rho, \lambda) \bigr| \\[6pt]
&\leq C_4 \exp \Bigl( \bigl[ \, \left| 2\pi \Im(\rho - \lambda) - \pi\nu_g \right| - 3\pi\nu_g \bigr] \| \bm{y}_{n - 1} \| \\[6pt]
&\hspace{1.5cm}+ \bigl[ \, \left| 2 \pi \Im(\rho - \lambda) - \pi\nu_g \right| - \pi\nu_g \bigr] | y_n |  \Bigr)
\end{aligned}
\end{equation}
with some $C_4(g, \bm{\omega}, \bm{x}_n, \bm{z}_{n - 1})$. Clearly, function from the right doesn't depend on $z_n$ and is integrable for $\Im(\rho - \lambda) \in (0, \nu_g)$. The same condition appeared when we considered the left-hand side of the equality \eqref{IV19}. Thus, assuming it we use dominated convergence theorem to evaluate the limit $z_n \rightarrow \infty$ of the integral \eqref{IV24}
\begin{equation}
\begin{aligned}
\lim_{z_n \rightarrow \infty} \, r( \bm{z}_n; \rho) \, Q_{}(&\bm{x}_n, \bm{z}_n; \rho, \lambda) = n \int_{\mathbb{R}^n} d\bm{y}_n \; e^{{2 \pi \imath}{} \left[ \left( \rho - \frac{\imath \hat{g}}{2} \right) (\bbx_n - \bby_n) + \lambda ( \bby_n - \bbz_{n - 1} ) \right] } \\[10pt]
&\times \K (\bm{x}_n, \bm{y}_{n - 1}) \, \mu  (\bm{y}_{n - 1}) \, \K (\bm{y}_{n - 1}, \bm{z}_{n - 1}) \, \K (y_n) \, \mu  (\bm{z}_{n - 1}).
\end{aligned}
\end{equation}
The integral over $y_n$ has separated and represents Fourier transform of function $\K $ \eqref{K-fourier}
\begin{equation}
\int_{\mathbb{R}} dy_n \; e^{{2 \pi \imath}{} \left( \lambda - \rho + \frac{\imath \hat{g}}{2} \right) y_n }  \K (y_n) = \sqrt{\omega_1 \omega_2} \, S(g) \, \KK \Bigl(\lambda - \rho + \frac{\imath \hat{g}}{2} \Bigr).
\end{equation}
The rest part of the integral coincides with the kernel of operator $\Lambda_{n}(\rho - \imath \hat{g}/2) \, Q_{n - 1}(\lambda)$ up to constant $d^2_{n - 1}(g|\bo)$.

So, starting with the commutativity of $Q$-operators in the limit $z_n \rightarrow \infty$ we obtain the equality
\begin{equation}
Q_n(\lambda) \, \Lambda_{n}\Bigl(\rho - \frac{\imath \hat{g}}{2} \Bigr) = \KK \Bigl(\lambda - \rho + \frac{\imath \hat{g}}{2} \Bigr) \Lambda_{n}\Bigl(\rho - \frac{\imath \hat{g}}{2} \Bigr) Q_{n-1}(\lambda),
\end{equation}
assuming $\Im(\rho - \lambda) \in (0, \nu_g) = (0, \Re \hat{g})$. Shifting the parameter $\rho \rightarrow \rho +  \imath \hat{g}/2$ we arrive at the identity stated in the theorem \eqref{QLcomm}.
\end{proof}
\begin{remark}
Exchange relation \eqref{QLcomm} at the edge of convergence strip $\Im(\lambda - \rho) = \pm \nu_g/2$ could be also obtained, although it is beyond the scope of the present paper. In this case one should think of the corresponding kernels (that won't be represented by absolutely convergent integrals anymore) as generalized functions.
\end{remark}

\setcounter{equation}{0}
\section{Eigenfunctions}

\subsection{Eigenfunctions of Macdonald operators}\label{sec-Macdonald}
Here we  prove that the wave functions \rf{I12} are eigenfunctions of Macdonald operators for complex valued periods $\bo$ and coupling constant $g$ assuming the conditions \rf{I0a},  \rf{I0b} and \rf{I13a}. The proof differs from that of \cite{HR1} by complex valued periods and the use of bilinear scalar product \rf{BM3a}.
\theor*
\begin{proof} We prove \rf{I13} by induction over $n$. It is clearly correct for $n=1$. For $n>1$ under the conditions  \rf{I0a} and  \rf{I0b} the function $\Psi_{\bl_n}(\bx_n)$ with
 \beq | \Im x_j | < \frac{g^*}{2}, \qquad | \Im(\l_j - \l_k) | < \theta, \qquad j,k=1,\ldots,n\eeq
with some positive $\theta$  is defined by absolutely converging multiple integral \rf{I12} (see Proposition \ref{Prop-Qconv} in Section \ref{sec-Qeigen}), which can be written via recursive procedure
	\begin{equation}\label{Psi-rec2a}
		\Psi_{\bm{\lambda}_n}(\bm{x}_n) = \Lambda_{n}(\lambda_n) \, \Psi_{\bm{\lambda}_{n - 1}}(\bm{x}_{n - 1}), \qquad \Psi_{\lambda_1}(x_1) = e^{{2\pi \imath}{} \lambda_1 x_1}.
	\end{equation}
	 In the space of functions $\vf(\by_{n-1})$ analytical in a small strip around the real plane
	introduce the symmetric bilinear pairing
	\beq\label{BM3a} (\vf,\psi)= \int_{\R^n}d\by_n \, \mu(\by_n) \, \vf(\by_n)\psi(-\by_n)\eeq
	assuming that it is defined when the corresponding integral converges.
	Denote by $\tau_y$ the operator that changes the sign of argument in a function
	\beqq  \tau_y \vf(\by_n)=\vf(-\by_n). \eeqq
	Then we can rewrite this pairing as
	\beq\label{BM4} (\vf,\psi)= \int_{\R^n}d\by_n \, \mu(\by_n) \, \vf(\by_n) \,\tau_y \bigl[\psi(\by_n) \bigr].\eeq
	In this notation, the inductive procedure \rf{Psi-rec2a} can be written as
	\beq\label{BM3b}	\Psi_{\bm{\lambda}_n}(\bm{x}_n)=e^{\const \l_n\bbx_n}\left(K(\bx_n,\by_{n-1}),e^{\const \l_n\bby_{n-1}}\,\tau_y \Psi_{\bm{\lambda}_{n-1}}(\bm{y}_{n-1})\right).\eeq
	Note the important property of the integration contour $C= \R^n$ in the integral \rf{BM3b}:
	it separates two series of poles of the kernel function:
	\beq\label{t04}
	\imath y_i=\imath x_j+\frac{g^\ast}{2} +m\o_1+k\o_2,\qquad \text{and}\qquad
	\imath y_i= \imath x_j-\frac{g^\ast}{2} -m\o_1-k\o_2, \qquad \ m,k\geq 0,
	\eeq
	and two series of poles of the measure function
	\beq \label{t04a}\imath y_i=\imath y_j+g  +m\o_1+k\o_2, \qquad\text{and}\qquad
	\imath y_i=\imath y_j-g  -m\o_1-k\o_2,\qquad \ m,k\geq 0.
	\eeq
		The kernel function $K(\bx_n,\by_n)$ satisfies the relation \cite{R2}
	\beq \label{I8} \left(M_r(\bx_n;g|\bo)- M_r(-\by_n;g|\bo)\right)K(\bx_n,\by_n)=0\eeq
	which is the corollary of the trigonometric version of the kernel function identity \cite{R2,KN}, valid for any tuples $\bx_n$ and $\by_n$ of $n$ complex variables and arbitrary parameter~$\a$:
	\beq\begin{split}\label{I9}
		\sum_{\substack{I_r\subset[n] \\ |I_r|=r}}\prod_{i\in I_r}\left(\prod_{j\in [n]\setminus I_r}\frac{\sin(x_i-x_j-\a)}{\sin(x_i-x_j)}\prod_{a=1}^{n}\frac{\sin(x_i-y_a+\a)}{\sin(x_i-y_a)}\right)
		=\\ \sum_{\substack{A_r\subset[n] \\ |A_r|=r}}\prod_{a\in A_r}\left(\prod_{b\in [n]\setminus A_r}
		\frac{\sin(y_a-y_b+\a)}{\sin(y_a-y_b)}\prod_{i=1}^{n}\frac{\sin(x_i-y_a+\a)}{\sin(x_i-y_a)}\right)
	\end{split}\eeq
In the limit $y_n$ tends to $\imath\infty$ the identity   \rf{I9} degenerates to
\beq\label{I9b}
\begin{split}
	\sum_{\substack{I_r\subset[n] \\ |I_r|=r}}\prod_{i\in I_r}\left(\prod_{j\in [n]\setminus I_r}\frac{\sin(x_i-x_j-\a)}{\sin(x_i-x_j)}\prod_{a=1}^{n-1}\frac{\sin(x_i-y_a+\a)}{\sin(x_i-y_a)}\right)
	=\\ \sum_{\substack{A_r\subset[n-1] \\  |A_r|=r}}\prod_{a\in A_r}\left(\prod_{b\in [n-1]\setminus A_r}
	\frac{\sin(y_a-y_b+\a)}{\sin(y_a-y_b)}\prod_{i=1}^{n}\frac{\sin(x_i-y_a+\a)}{\sin(x_i-y_a)}\right)+\\
	\sum_{\substack{A_{r-1}\subset[n-1] \\ |A_{r-1}|=r-1}}\prod_{a\in A_{r-1}}\left(\prod_{b\in [n-1]\setminus A_{r-1}}
	\frac{\sin(y_a-y_b+\a)}{\sin(y_a-y_b)}\prod_{i=1}^{n}\frac{\sin(x_i-y_a+\a)}{\sin(x_i-y_a)}\right)	
\end{split}\eeq
which implies the relation
\beq \label{I8a}
\begin{aligned}
	M_r(\bx_n;g|&\bo)K(\bx_n,\by_{n-1}) \\[5pt]
	&= \Bigl[M_r(-\by_{n-1};g|\bo)+M_{r-1}(-\by_{n-1};g|\bo) \Bigr] K(\bx_n,\by_{n-1}).
\end{aligned}
\eeq
Besides, Macdonald operators $M_r(\by_{n-1};g|\bo)$ are symmetric with respect to the bilinear pairing
	\rf{BM3a},	 see \cite[Chapter VI, §9, eq.(9.4)]{M},
	\beq\label{BM6} \Bigl(M_r(\by_{n-1};g|\bo)\vf(\by_{n-1}),\psi(\by_{n-1})\Bigr)=
	\Bigl(\vf(\by_{n-1}),M_r(\by_{n-1};g|\bo)\psi(\by_{n-1})\Bigr).
	\eeq
	Acting with Macdonald operator $M_r(\bx_n;g|\bo)$ on the wave function \eqref{BM3b} we have
	\beq\label{BM7}\begin{split} & M_r(\bx_n;g|\bo) \, e^{\const \l_n\bbx_n}\left(K(\bx_n,\by_{n-1}),e^{\const \l_n\bby_{n-1}}\,\tau_y \Psi_{\bm{\lambda}_{n-1}}(\bm{y}_{n-1})\right)\\[5pt]
		&= e^{{2\pi \imath \l_n}\bbx_n + 2\pi r \l_n \o_1} \Bigl(M_r(\bx_n;g|\bo) K(\bx_n,\by_{n-1}),\ e^{{2\pi \imath \l_n}\bby_{n-1}} \, \tau_y \Psi_{\bl_{n-1}}(\by_{n-1})\Bigr).\end{split}	\eeq
	Using \rf{I8a} we rewrite the right hand side of \rf{BM7} as
	\beq\label{BM7a}\begin{split}
	&e^{{2\pi \imath \l_n}\bbx_n + 2\pi r \l_n \o_1} \biggl[\Bigl(\tau_y M_r(\by_{n-1};g|\bo)\tau_y K(\bx_n,\by_{n-1}),\ e^{{2\pi \imath \l_n}\bby_{n-1}} \, \tau_y \Psi_{\bl_{n-1}}(\by_{n-1}) \Bigr)\\[4pt]
	&\qquad+\Bigl(\tau_y M_{r-1}(\by_{n-1};g|\bo)\tau_y K(\bx_n,\by_{n-1}),\ e^{{2\pi \imath \l_n}\bby_{n-1}} \, \tau_y \Psi_{\bl_{n-1}}(\by_{n-1}) \Bigr)\biggr]\\[4pt]
	=\, &e^{{2\pi \imath \l_n}\bbx_n + 2\pi r \l_n \o_1} \biggl[\Bigl(M_r(\by_{n-1};g|\bo)\tau_y K(\bx_n,\by_{n-1}),\ e^{-{2\pi \imath \l_n}\bby_{n-1}} \, \Psi_{\bl_{n-1}}(\by_{n-1}) \Bigr)\\[4pt]
	&\qquad+\Bigl (M_{r-1}(\by_{n-1};g|\bo)\tau_y K(\bx_n,\by_{n-1}),\ e^{-{2\pi \imath \l_n}\bby_{n-1}} \,  \Psi_{\bl_{n-1}}(\by_{n-1}) \Bigr)\biggr].
	\end{split}\eeq
Then using the symmetricity of Macdonald operators \rf{BM6} we rewrite \rf{BM7a} as
\beq\label{BM5}\begin{split}
	e^{{2\pi \imath \l_n}\bbx_n + 2\pi r \l_n \o_1} \biggl[\Bigl(\tau_y K(\bx_n,\by_{n-1}),\  M_r(\by_{n-1};g|\bo)e^{{-2\pi \imath \l_n}\bby_{n-1}} \, \Psi_{\bl_{n-1}}(\by_{n-1})  \Bigr)\\ + \Bigl( \tau_y K(\bx_n,\by_{n-1}),\ M_{r-1}(\by_{n-1};g|\bo)e^{{-2\pi \imath \l_n}\bby_{n-1}} \, \Psi_{\bl_{n-1}}(\by_{n-1})  \Bigr)\biggr]
		\end{split}\eeq
	under the assumption that the shifts of the integration contour implemented in the equality \rf{BM6} do not spoil separation of the poles \rf{t04} and \rf{t04a} property. For \rf{t04} this happens when
	 \beq\label{BM5a} \Re \o_1<\Re g^\ast,\eeq
	 which is equivalent to \eqref{I13a}. The conditions \rf{t04a} could be violated in the right hand side of the adjointness relation. However, in this case the nearest poles $\imath y_a=\imath y_b+g$ are canceled by zeroes of the coefficients $\sh\frac{\pi}{\o_2} (y_a-y_b-\imath g)$ of Macdonald operators, see \cite[Section 2]{BDKK}, and we do not get here additional restrictions.
	  Finally we rewrite  \rf{BM5} as
	  \beq\label{BM5c}\begin{split}
	  	&e^{2\pi \imath \l_n\bbx_n} \biggl[ \Bigl( \tau_y K(\bx_n,\by_{n-1}),\  e^{-2\pi \imath \l_n\bby_{n-1}} \,M_r(\by_{n-1};g|\bo) \Psi_{\bl_{n-1}}(\by_{n-1})  \Bigr)\\
	  	&\qquad+\Bigl( \tau_y K(\bx_n,\by_{n-1}),e^{{-2\pi \imath \l_n}\bby_{n-1}+2\pi \l_n\o_1} M_{r-1}(\by_{n-1};g|\bo) \, \Psi_{\bl_{n-1}}(\by_{n-1}) \Bigr)\biggr]
	  	\\
		=\,&e^{2\pi \imath \l_n\bbx_n} \biggl[ \Bigl(  K(\bx_n,\by_{n-1}),\  e^{2\pi \imath \l_n\bby_{n-1}} \,\tau_y M_r(\by_{n-1};g|\bo) \Psi_{\bl_{n-1}}(\by_{n-1})  \Bigr)\\
		&\qquad+\Bigl( K(\bx_n,\by_{n-1}),e^{{2\pi \imath \l_n}\bby_{n-1}+2\pi \l_n\o_1} \tau_y M_{r-1}(\by_{n-1};g|\bo) \, \Psi_{\bl_{n-1}}(\by_{n-1}) \Bigr)\biggr].
	  \end{split}\eeq
  By induction assumption,
 \beq\label{BM5d}  M_{s}(\by_{n-1};g|\bo)\Psi_{\bl_{n-1}}(\by_{n-1}) =e_s\big(e^{2\pi  \l_1\o_1},\ldots, e^{2\pi  \l_{n-1}\o_1}\big) \Psi_{\bl_{n-1}}(\by_{n-1}).\eeq
  Then \rf{BM5c} implies the relation
	\beq\label{BM5e} \begin{split}&M_{r}(\bx_{n};g|\bo)\Psi_{\bl_{n}}(\bx_{n}) = \\[6pt]
	&\quad  \Bigl[e_r\big(e^{2\pi \l_1\o_1},\ldots, e^{2\pi  \l_{n-1}\o_1}\big)+ e^{\const \l_n\o_1}e_{r-1}\big(e^{2\pi  \l_1\o_1},\ldots, e^{2\pi  \l_{n-1}\o_1}\big)\Bigr]\Psi_{\bl_{n}}(\bx_{n})\end{split}
	\eeq
which in turn implies the theorem statement \rf{I13} due to the relation  	
\beqq e_r(z_1,\ldots, z_n)=e_r(z_1,\ldots, z_{n-1})+z_ne_{r-1}(z_1,\ldots, z_{n-1}).\eeqq
\end{proof}

\subsection{Eigenfunctions of $Q$-operators}\label{sec-Qeigen}
In this section we show that the iterated integrals \rf{I12} represent as well joint eigenfunctions of the commutative family of operators $Q_n(\l)$. Modulo convergence of the appearing integrals, it is a formal consequence of Theorem \ref{theoremQL}. However, these convergence details are not trivial.
\QPsi*
Proof of this theorem follows from the iterative representation of the wave function
\begin{equation}\label{Psi-iter}
\Psi_{\bl_n}(\bx_n; g)=\LL{n}(\l_n) \, \LL{n-1}(\l_{n-1})\cdots \LL{2}(\l_2) \, e^{\const \l_1x_1}
\end{equation}
and the exchange relation given by Theorem \ref{theoremQL}
\begin{equation}
Q_n(\lambda) \, \Lambda_{n}(\rho) = \KK(\lambda - \rho) \; \Lambda_{n}(\rho) \, Q_{n-1}(\lambda).
\end{equation}
The subtle point is the possibility to switch integration order inside of multiple integrals appearing during this calculation. It is indeed possible due to the absolute convergence of these integrals, which we prove in the following proposition. After it we give a proof of Theorem \ref{thm-QPsi}.

\begin{proposition}\label{Prop-Qconv} Multiple integrals
\begin{align}\label{int1}
I_{\lambda, \bm{\lambda}_n} &= Q_n(\lambda) \, \LL{n}(\l_n) \, \LL{n-1}(\l_{n-1})\cdots \LL{2}(\l_2) \, e^{\const \l_1x_1}, \\[7pt]
\label{int2}
J_{\lambda, \bm{\lambda}_n} &= \Lambda_{n}(\lambda_n) \, Q_{n - 1}(\lambda) \, \LL{n-1}(\l_{n-1})\cdots \LL{2}(\l_2) \, e^{\const \l_1x_1}, \\[7pt]
\label{int3}
\Psi_{\bl_n} &= \LL{n}(\l_n) \, \LL{n-1}(\l_{n-1})\cdots \LL{2}(\l_2) \, e^{\const \l_1x_1}.
\end{align}
are absolutely convergent for $\lambda, \lambda_j \in \mathbb{C}$ such that
\begin{equation}
	 |\Im (\l-\l_n)|<\frac{\nu_g}{2}(1-\ve),\qquad
|  \Im(\lambda_k - \lambda_j)| < \theta(\ve), \qquad k,j = 1, \dots, n
\end{equation}
for any $\ve \in [0,1)$ and
\beq \theta(\ve) = \frac{\nu_g}{4 (n - 1)! e} \, \ve. \eeq
\end{proposition}
In particular, the wave function $\Psi_{\bl_n}(\bx_n)$ is given by a multiple integral, converging when all the variables are in small strips around the real lines.
\begin{proof}  We present the proof for $n\geq 2$. For $n=1$ it is an elementary check.
	
Consider the first integral \eqref{int1}. It has $n$ levels of integration variables. Denoting integration variables of $k$-th level by
\begin{equation}
\bm{y}_{k} = \bigl( y_1^{(k)} , \dots, y_k^{(k)} \bigl),
\end{equation}
we write integral in its full form
\begin{equation}
\begin{aligned}
I_{\lambda, \bm{\lambda}_n} (\bm{x}_n) = C_I \int_{\mathbb{R}^{n(n+1)/2} } d\bm{y}_n & d\bm{y}_{n - 1} \dots d\bm{y}_1 \; e^{ {2\pi \imath}{} \left[ \lambda \bbx_n + (\lambda_n - \lambda) \bby_{n} \right] } \, \K  (\bm{x}_n, \bm{y}_{n}) \\[5pt]
& \hspace{0.6cm} \times  \prod_{k = 2}^n e^{{2\pi \imath}{}(\lambda_{k - 1} - \lambda_{k}) \bby_{k - 1}} \, \mu ( \bm{y}_{k} ) \, \K  (\bm{y}_{k}, \bm{y}_{k - 1} ).
\end{aligned}
\end{equation}
Here $C_I$ contains all constants $d_k$ from kernels of operators. Denote the integrand as $F$.
Assume that
\beq\label{kh1}|\Im (\l-\l_n)| \leq \delta_Q\frac{\nu_g}{2},\qquad |\Im (\l_{k - 1}-\l_{k})| \leq \delta_\Lambda\frac{\nu_g}{2},\qquad 2<k\leq n\eeq
where $\delta_Q$ and $\delta_{\Lambda}$ are some positive constants. Using bounds \eqref{Kmu-bound}
\begin{equation}\label{Kmu-bounds}
| \mu  (x) | \leq C e^{ \pi\nu_g |x| }, \qquad |\K (x)| \leq C e^{ - \pi\nu_g |x| }, \qquad x \in \mathbb{R},
\end{equation}
together with triangle inequalities $|x_i - y_j| \geq |y_j| - |x_i|$ we get
\begin{equation}\label{int1b}
| F | \leq C_1 \exp \pi\nu_g\biggl( \bigl[\delta_Q - n   \bigr]  \|\bm{y}_n \| +  S_n (\bm{y}_1, \dots,  \bm{y}_n ) +\delta_{\Lambda}\sum_{k=1}^{n-1}\|\bm{y}_k\|\biggr),
\end{equation}
with some $C_1(g, \bm{\omega}, \bm{x}_n)$ and the function $S_n$ in exponent is defined by recurrence relation
\begin{equation}\label{Srec}
\begin{aligned}
S_n (\bm{y}_1, \dots, \bm{y}_n ) = \sum_{\substack{i, j = 1 \\ i \not= j}}^n \bigl| y_i^{(n)} - y_j^{(n)} \bigr| &- \sum_{i = 1}^{n} \sum_{j = 1}^{n - 1} \, \bigl| y_i^{(n)} - y_j^{(n - 1)} \bigr| \\
&+ S_{n - 1} (\bm{y}_1, \dots,  \bm{y}_{n - 1} )
\end{aligned}
\end{equation}
with $S_1 = 0$. As before, $\| \bm{y}_k \|$ denotes $L^1$-norm.  In appendix we prove the inequality~\eqref{Sest2}
\begin{equation}
S_n \leq (n - 1 +  \epsilon) \| \bm{y}_n \| - \frac{\ve}{c_n} \, \sum_{k = 1}^{n - 1} \| \bm{y}_k \|
\end{equation}
for any $\epsilon \in \left[ 0,2(n - 1) \right]$, where $c_n$ is defined by the recurrence relation
\beq\label{kh2}c_n = (n - 1)(c_{n-1} + 1)\qquad\text{with}\qquad c_1 = 0.\eeq
Clearly,
\beq\label{kh5} (n-1)! \leq c_n<(n-1)!e, \qquad n\geq 2.\eeq
Inserting it in \eqref{int1b} we obtain
\begin{equation}\label{kh3}
|F| \leq C_1 \exp \pi\nu_g \biggl( \Bigl[ ( -1+\delta_Q + \ve)  \Bigr]  \|\bm{y}_n \|  +\left(\delta_{\Lambda}-\frac{\ve}{c_n}\right)  \sum_{k = 1}^{n - 1} \| \bm{y}_k \| \biggr).
\end{equation}
We have an exponential decay in the right hand side of \rf{kh3}, if
\beq\label{kh4} \delta_Q<1-\ve,\qquad \delta_{\Lambda}<\frac{\ve}{(n-1)!e}<\frac{\ve}{c_n}.
\eeq
This proves the statement of the proposition for the integral \rf{int1}.

Next, consider the second integral \eqref{int2}. It also has $n$ levels of integration variables. We will denote by $\bm{t}_{n - 1} = (t_1, \dots, t_{n - 1})$ variables associated with the operator $\Lambda_{n}$ and the rest levels of variables, as before, will be denoted by $\bm{y}_{k}$ ($k = 1, \dots, n - 1$). Full form of the integral
\begin{align} \nonumber
J_{\lambda, \bm{\lambda}_n} (\bm{x}_n) &= C_J \int  d\bm{t}_{n - 1} d\bm{y}_{n - 1} \ldots d\bm{y}_1 \, e^{ {2\pi \imath}{} \left[ \lambda_n \bbx_n + (\lambda - \lambda_n) \bbt_{n - 1} + (\lambda_{n - 1} - \lambda) \bby_{n - 1} \right] } \, \K (\bm{x}_n, \bm{t}_{n - 1})  \\[4pt]
& \times \mu  (\bm{t}_{n - 1}) \, \K (\bm{t}_{n - 1}, \bm{y}_{n - 1} ) \prod_{k = 2}^{n - 1} e^{{2\pi \imath}{}  (\lambda_{k - 1} - \lambda_{k}) \bby_{k - 1}} \, \mu   ( \bm{y}_{k}) \, \K  (\bm{y}_{k}, \bm{y}_{k - 1}).
\end{align}
Again, $C_J$ contains all unimportant constants. The integration goes over $\mathbb{R}^{(n - 1)(n + 2)/2}$. Denote the integrand as $G$. Assuming the conditions \rf{kh1}, using the bounds \eqref{Kmu-bounds} and triangle inequalities we arrive at
\begin{equation}\label{int2b}
| G | \leq  C_2 \exp \pi\nu_g\biggl(\delta_Q \, \left| \sum\limits_{j = 1}^{n - 1} \bigl( t_j - y_j^{(n-1)} \bigr) \right|+\delta_\Lambda \sum_{k = 1}^{n - 1} \| \bm{y}_k \|- n   \| \bm{t}_{n - 1} \| +  T_{n - 1} \biggr) ,
\end{equation}
with some $C_2(g, \bm{\omega}, \bm{x}_n)$, where the function $T_{n - 1}$ is defined as
\begin{equation}
\begin{aligned}
T_{n - 1}(\bm{y}_1, \dots, \bm{y}_{n - 1}, \bm{t}_{n - 1}) = \sum_{\substack{i, j = 1 \\ i \not= j}}^{n - 1} | t_i - t_j | &- \sum_{i, j = 1}^{n - 1} \bigl| t_i - y_j^{(n - 1)} \bigr| \\[4pt]
&+ S_{n - 1}(\bm{y}_1, \dots, \bm{y}_{n - 1}).
\end{aligned}
\end{equation}
In appendix we prove the inequality \eqref{Test} (here we replace $r$ by $1-\ve$)
\begin{equation}\label{kh6}
\begin{aligned}
T_{n - 1} \leq ( n -  \ve ) \| \bm{t}_{n - 1} \| 
- \frac{ \ve}{2 (n - 1) c_{n - 1}} \sum_{k = 1}^{n - 1} \| \bm{y}_k \|
- (1-\ve) \, \biggl| \sum_{j = 1}^{n - 1} \bigl( t_j - y_j^{(n - 1)} \bigr) \biggr|
\end{aligned}
\end{equation}
for any $\ve \in [0, 1]$. Using it we have the bound
\beqq\begin{aligned}
| G | \leq C_2 \exp \pi\nu_g\biggl(   -\ve \| \bm{t}_{n - 1} \| &+\left(\delta_\Lambda - \frac{ \ve}{2 (n - 1) c_{n - 1}}\right) \sum\limits_{k = 1}^{n - 1} \| \bm{y}_k \| \\[5pt]
 &+ (\delta_Q-1+ \ve) \biggl| \sum\limits_{j = 1}^{n - 1} \bigl( t_j - y_j^{(n-1)} \bigr) \biggr| \, \biggr) .
\end{aligned}\eeqq
We thus see that the integral \eqref{int2} converges once
\beq\label{kh7}
\delta_Q<1-\ve, \qquad \delta_{\Lambda}< \frac{ \ve}{2 (n-1)! e}   < \frac{ \ve}{2 (n - 1) c_{n - 1}}.
\eeq

Finally, the third integral \eqref{int3} in its full form
\begin{equation}
\begin{aligned}
\Psi_{\bm{\lambda}_n} (\bm{x}_n) = C_\Psi \int_{\mathbb{R}^{n(n-1)/2} } & d\bm{y}_{n - 1} \dots d\bm{y}_1 \; e^{ {2\pi \imath}{} \left[ \lambda_n \bbx_n + (\lambda_{n - 1} - \lambda_{n}) \bby_{n - 1} \right] } \K (\bm{x}_n, \bm{y}_{n - 1}) \\[5pt]
& \hspace{0.6cm} \times  \prod_{k = 2}^{n - 1} e^{{2\pi \imath}{}(\lambda_{k - 1} - \lambda_{k}) \bby_{k - 1}} \, \mu   ( \bm{y}_{k} ) \, \K  (\bm{y}_{k}, \bm{y}_{k - 1} )
\end{aligned}
\end{equation}
almost coincides with the first integral $I_{\lambda_n, \bm{\lambda}_{n - 1}}$ up to additional functions
\begin{equation}
e^{{2\pi \imath}{}  \lambda_n x_n} \,  \prod_{j = 1}^{n - 1} \K (x_n - y_j)
\end{equation}
that only improve convergence.
Its integrand, denote it by $H$, can be bounded in a similar way as
\begin{equation}\label{kh8}
	|H| \leq C_3 \exp \pi\nu_g \biggl( \Bigl[ ( -2+\delta_\Lambda + \ve)  \Bigr]  \|\bm{y}_n \|  +\left(\delta_{\Lambda}-\frac{\ve}{c_n}\right)  \sum_{k = 1}^{n - 1} \| \bm{y}_k \| \biggr).
\end{equation}
 It converges if
\beq\label{kh9} \delta_{\Lambda}<2-\ve,\qquad \delta_{\Lambda}<\frac{\ve}{(n-1)!e}
\eeq
for some $\ve >0$. Assuming $\ve \in [0, 1)$ the first condition follows from the second one.

So, all three integrals converge under assumptions
\begin{equation}
|\Im(\l - \l_n) | < \frac{\nu_g}{2}(1 - \ve), \qquad |\Im(\l_j - \l_k)|< \frac{\nu_g}{2} \cdot  \frac{\ve}{2(n - 1)! e}
\end{equation}
with any $\ve \in [0,1)$.
\end{proof}

\begin{remark}\label{remark-3}
Since all constants that appear in the bounds are uniform on compact sets of parameters $\bo, g$ preserving the conditions \eqref{I0a}, \eqref{I0b}, the considered multiple integrals are analytic functions of the parameters $\bo, g$, as well as of the variables $x_i$ and $\l_j$. In the case of real periods $\o_1, \o_2$ convergence of the third integral \eqref{int3}, that is the wave function representation, was first proved in~\cite{HR1}.
\end{remark}

\begin{proof}[Proof of Theorem \ref{thm-QPsi}]
In the notation of the previous proposition the theorem states that
\begin{equation}
I_{\lambda, \bm{\lambda}_n} = \prod_{j = 1}^n \KK (\lambda - \lambda_j) \, \Psi_{\bm{\lambda}_n}.
\end{equation}
The proof goes by induction. In the case $n = 1$ this integral
\begin{equation}
I_{\lambda, \lambda_1} = d_1(g|\bo) \int_{\mathbb{R}} dy_1 \, e^{ {2\pi \imath}{} \left[ \lambda x_1 + (\lambda_1 - \lambda) y_1  \right] } \K (x_1 - y_1) = \KK(\lambda - \lambda_1) \, e^{ {2\pi \imath}{} \lambda_1 x_1 }
\end{equation}
coincides with the Fourier transform of the function $\K $ \eqref{K-fourier}. Next, consider the case of arbitrary $n$. The main ingredient of the proof is the local relation \eqref{QLcomm}
\begin{equation}\label{QL}
Q_n(\lambda) \, \Lambda_{n}(\rho) = \KK(\lambda - \rho) \;\Lambda_{n}(\rho) \, Q_{n-1}(\lambda).
\end{equation}
Consider the function in question
\begin{equation}
I_{\lambda, \bm{\lambda}_n} = Q_n (\lambda) \, \Psi_{\bm{\lambda}_n} = Q_n(\lambda) \,\Lambda_{n} (\lambda_n) \cdots \Lambda_{2}(\lambda_2) \,  \Psi_{\lambda_1}.
\end{equation}
This integral has $n$ levels of integration variables that we denoted by
\begin{equation}
\bm{y}_k = \bigl(y_1^{(k)}, \dots, y_k^{(k)} \bigr), \qquad k = 1, \dots, n.
\end{equation}
By Proposition \ref{Prop-Qconv} the integral is absolutely convergent. Therefore, we can use Fubini's theorem and evaluate integrals in any order. Take integral over $\bm{y}_n$ to be the first in this order. Then we use the local relation \eqref{QL} and obtain
\begin{equation}
\begin{aligned}
I_{\lambda, \bm{\lambda}_n} &= \KK(\lambda - \lambda_n) \;\Lambda_{n}(\lambda_n) \, Q_{n - 1}(\lambda) \, \Lambda_{n - 1}(\lambda_{n - 1}) \cdots  \Lambda_{2}(\lambda_2) \, \Psi_{\lambda_1} \\[4pt]
&= \KK(\lambda - \lambda_n) \; J_{\lambda, \bm{\lambda}_n}.
\end{aligned}
\end{equation}
Obtained integral coincides with the function $J_{\lambda, \bm{\lambda}_n}$ from Proposition \ref{Prop-Qconv}. This integral has $n$ levels of integration variables that we denoted by $\bm{t}_{n - 1}$ (variables associated with the operator $\Lambda_{n}$ from the left) and $\bm{y}_{k}$ (the rest ones, $k = 1, \dots, n - 1$). Since by proposition it is also absolutely convergent we can change order of integrals and make the integral over $\bm{t}_{n - 1}$ to be the last one. Then the multiple integral we should take before it
\begin{equation}
I_{\lambda, \bm{\lambda}_{n - 1}} = Q_{n - 1}(\lambda) \, \Lambda_{n - 1}(\lambda_{n - 1}) \cdots  \Lambda_{2}(\lambda_2) \, \Psi_{\lambda_1}
\end{equation}
represents the function on the previous step of induction. So, to conclude the proof we use induction assumption and obtain desired formula.
\end{proof}

\setcounter{equation}{0}
\section{Duality}\label{sec-duality}
In this section, we prove the equivalence of two integral representations of the wave function
$\Psi_{\bm{\lambda}_{n}}(\bm{x}_{n})$ encoded in duality relation \rf{duality}. This will be proved in two steps. First, using the spectral properties of operators $Q_n(\l)$ and $\QQ_n(x)$ proved in the previous section and the bound on the wave function from \cite{HR3}, we establish the relation \rf{duality} for the case of real positive periods $\o_1,\o_2$. Then using analiticity of the wave function with respect to the periods (see Remark \ref{remark-3} in the previous section) we extend this result to complex periods.
\duality*

 We note that functions from both sides exist under assumptions \eqref{gg-ineq}
\begin{equation}
0< \Re g < \Re \o_1 + \Re \o_2, \qquad 0 < \Re \hat{g}^* < \Re \hat{\o}_1 + \Re \hat{\o}_2.
\end{equation}\cbk
 We also assume $x_j, \l_j \in \R$; the duality relation for complex variables then follows from the analyticity arguments.  To prove this theorem, firstly, note that kernels of the $Q$-operator \eqref{Qker} and raising $\Lambda$-operator \eqref{Lker} are connected by the formula
\begin{equation}
\Lambda(\bm{x}_n, \bm{y}_{n - 1}; \lambda) = e^{ {2 \pi \imath}{} \lambda x_n } \, Q(\bm{x}_{n - 1}, \bm{y}_{n - 1}; \lambda) \, \prod_{j = 1}^{n - 1} \K (x_n - y_j).
\end{equation}
Therefore, the operators themselves are connected in the following way
\begin{equation}\label{LQ}
\Lambda_{n}(\lambda) = e^{{2 \pi \imath}{} \lambda x_n } \, Q_{n-1}(\lambda) \, \prod_{j = 1}^{n - 1} \K (x_n - x_j).
\end{equation}

Secondly, the function 	$\hat{\Psi}_{\bm{x}_n}(\bm{\lambda}_n)=\Psi_{\bm{x}_n}(\bm{\lambda}_n; \gg|\hat{\bo})	$ from the right-hand side of duality relation \eqref{duality} admits an analogous integral representation
\begin{equation}
\hat{\Psi}_{\bm{x}_n}(\bm{\lambda}_n) = \LLL_n(x_n) \, \hat{\Psi}_{\bm{x}_{n - 1}}(\bm{\lambda}_{n - 1}),
\end{equation}
iterative with respect to $\lambda_j$ variables
\begin{equation}\label{d1}
\bigl(\LLL_n(x) f\bigr)(\bm{\lambda}_n) = d_{n - 1}(\hat{g}^*|\hat{\bo}) \int_{\mathbb{R}^{n - 1}} d\bm{\gamma}_{n - 1} \, \LLL(\bm{\lambda}_n, \bm{\gamma}_{n - 1}; x) f(\bm{\gamma}_{n - 1})
\end{equation}
where the kernel
\begin{equation}
\LLL(\bm{\lambda}_n, \bm{\gamma}_{n - 1}; x) = e^{\const x (\bbl_n - \bbg_{n - 1})} \KK(\bl_n, \bg_{n - 1}) \hat{\mu}(\bg_{n - 1}).
\end{equation}
It is also an eigenfunction of the $Q$-operator acting on functions of $\lambda_j$
\begin{equation}\label{d2}
\bigl(\QQ_n(x) f\bigr)(\bm{\lambda}_n) = d_{n}(\hat{g}^*|\hat{\bo}) \int_{\mathbb{R}^n} d\bm{\gamma}_{n} \, \QQ(\bm{\lambda}_n, \bm{\gamma}_{n}; x) f(\bm{\gamma}_{n})
\end{equation}
with the kernel
\begin{equation}
\hat{Q}(\bm{\lambda}_n, \bm{\gamma}_{n}; x) = e^{\const x (\bbl_n - \bbg_{n})} \KK(\bl_n, \bg_{n}) \hat{\mu}(\bg_{n})
\end{equation}
and the corresponding eigenvalue
\begin{equation}\label{Q'Psi}
\QQ_{n}(x) \, \hat{\Psi}_{ \bm{x}_n }(\bm{\lambda}_n) = \prod_{j = 1}^n \K(x-x_j) \, \hat{\Psi}_{ \bm{x}_n }(\bm{\lambda}_n ).
\end{equation}
All of this is just a matter of renaming the variables $\lambda_j \leftrightarrow x_j$ and parameters $\bo\leftrightarrow \hat{\bo}$, $g \leftrightarrow \gg$, see the definitions of the functions $\KK$ \eqref{I6c} and $\hat{\mu}$ \eqref{I5d}. The  operators \rf{d1} and \rf{d2} \cbk are also connected by the formula analogous to \eqref{LQ}
\begin{equation}\label{LQ'}
\LLL_{n}(x) = e^{ {2 \pi \imath}{} x \lambda_n } \, \QQ_{n - 1}(x) \, \prod_{j = 1}^{n - 1} \KK (\lambda_n - \lambda_j).
\end{equation}

Thirdly, in the proof we use the bound derived in \cite[Theorem 2.5]{HR3}. This bound holds when $x_j \in \mathbb{C}$ and $\Re g \in (0, \max(\omega_1, \omega_2)]$, but we only need its specialization for the case of real $x_j$. Denote
\begin{equation}\label{mu'}
\mu'(\bx_n) = \prod_{\substack{i,j=1 \\ i< j}}^n \mu(x_i - x_j), \qquad \hat{\mu}'(\bl_n) = \prod_{\substack{i,j=1 \\ i< j}}^n \hat{\mu}(\l_i - \l_j),
\end{equation}
so that comparing with \eqref{I5} and \eqref{I6d} we have
\begin{equation}\label{mumu'}
\mu(\bm{x}_n) = \mu'(\bx_n) \mu'(-\bx_n), \qquad  \hat{\mu}(\bl_n) = \hat{\mu}'(\bl_n) \hat{\mu}'(-\bl_n).
\end{equation}
\begin{theorem*}
\cite[Specialization of Theorem 2.5]{HR3} Assuming $\Re g \in (0, \max(\omega_1, \omega_2)]$, the wave function admits the bound
\begin{equation}\label{Psi-HRbound}
\Bigl| \mu'(\bm{x}_n)\, \hat{\mu}'(\bm{\lambda}_n) \, \Psi_{ \bm{\lambda}_n }(\bm{x}_n) \Bigr| \leq P(\bm{x}_n; g| \bo),
\end{equation}
where $P$ is a polynomial with $\deg P \leq n(n - 1)/2$.
\end{theorem*} \cbk

\begin{proof}[Proof of Theorem \ref{thm-duality}]
At first, assume $\o_1, \o_2$ are real positive. We prove the theorem by induction. The case $n = 1$
\begin{equation}
\Psi_{\lambda_1}(x_1) = \hat{\Psi}_{x_1}(\lambda_1) = e^{ {2 \pi \imath}{} \lambda_1 x_1 }
\end{equation}
is obvious. Suppose we proved the relation for $n - 1$ case
\begin{equation}\label{Psi-ind}
\Psi_{\bm{\lambda}_{n - 1}}(\bm{x}_{n - 1}) = \hat{\Psi}_{\bm{x}_{n - 1}}(\bm{\lambda}_{n - 1}).
\end{equation}
Using formula \eqref{LQ} we write the wave function recursively with the help of $Q$-operator
\begin{equation}
\Psi_{\bm{\lambda}_n}(\bm{x}_n) = e^{ {2 \pi \imath}{} \lambda_n x_n } \, Q_{n-1}(\lambda_n) \, \prod_{j = 1}^{n - 1} \K (x_n - x_j) \, \Psi_{\bm{\lambda}_{n - 1}}(\bm{x}_{n - 1}).
\end{equation}
By induction assumption \eqref{Psi-ind} we can represent $\K $ functions here as action of the operator $\QQ_{n-1}(x_n)$, see eigenvalue identity \eqref{Q'Psi},
\begin{equation}\label{PsiQQ'}
\Psi_{\bm{\lambda}_n}(\bm{x}_n) = e^{ {2 \pi \imath}{} \lambda_n x_n } \, Q_{n-1}(\lambda_n) \, \QQ_{n - 1}(x_n) \, \Psi_{\bm{\lambda}_{n - 1}}(\bm{x}_{n - 1}).
\end{equation}
Since operators $Q_{n-1}(\lambda_n)$ and $\QQ_{n-1}(x_n)$ act on different types of variables ($\bm{x}_{n - 1}$ and $\bm{\lambda}_{n - 1}$ correspondingly), they commute: switching their order is equivalent to switching order of the corresponding integrals. Let us postpone convergence arguments to the end of this proof and assume that we can switch their order. Then we arrive at
\begin{equation}
\Psi_{\bm{\lambda}_n}(\bm{x}_n) = e^{ {2 \pi \imath}{} \lambda_n x_n } \, \QQ_{n - 1}(x_n) \, Q_{n-1}(\lambda_n)  \, \Psi_{\bm{\lambda}_{n - 1}}(\bm{x}_{n - 1}).
\end{equation}
Next step is to use the eigenvalue identity \eqref{QPsi} for the operator $Q_{n-1}(\lambda_n)$
\begin{equation}
\Psi_{\bm{\lambda}_n}(\bm{x}_n) = e^{ {2 \pi \imath}{} \lambda_n x_n } \, \QQ_{n - 1}(x_n) \, \prod_{j =1}^{n - 1} \KK (\lambda_n - \lambda_j)  \, \Psi_{\bm{\lambda}_{n - 1}}(\bm{x}_{n - 1}).
\end{equation}
Finally, using induction assumption \eqref{Psi-ind} and formula \eqref{LQ'} we obtain function $\hat{\Psi}_{\bm{x}_n}(\bm{\lambda}_n)$ from the right.

It is left to prove that we can switch integration order in representation \eqref{PsiQQ'}. Let us write it explicitly
\begin{equation}
\begin{aligned}
\Psi_{\bm{\lambda}_n}(\bm{x}_n) = e^{ {2 \pi \imath}{} \lambda_n x_n } \int_{\mathbb{R}^{n - 1}} d\bm{y}_{n - 1} \, \mu  (\bm{y}_{n - 1}) \, e^{ {2 \pi \imath}{} \lambda_n (\bm{x}_{n - 1} - \bm{y}_{n - 1} ) } \, \K (\bm{x}_{n - 1}, \bm{y}_{n - 1}) \\[5pt]
\times \int_{\mathbb{R}^{n - 1}} d\bm{\gamma}_{n - 1} \, \hat{\mu}_{}(\bm{\gamma}_{n - 1}) \, e^{ {2 \pi \imath}{} x_n (\bm{\lambda}_{n - 1} - \bm{\gamma}_{n - 1} ) } \, \KK (\bm{\lambda}_{n - 1}, \bm{\gamma}_{n - 1}) \; \Psi_{\bm{\gamma}_{n - 1}} (\bm{y}_{n - 1}).
\end{aligned}
\end{equation}
To switch order of integrals over $\bm{y}_{n - 1}$ and $\bm{\gamma}_{n - 1}$ we bound integrand and use Fubini's theorem. First, factor measure functions
\begin{equation}
\mu(\by_{n - 1}) = \mu'(\by_{n - 1}) \mu'(-\by_{n - 1}), \qquad \hat{\mu}(\bg_{n - 1}) = \hat{\mu}'(\bg_{n - 1}) \hat{\mu}'(-\bg_{n - 1}),
\end{equation}
where $\mu'$ and $\hat{\mu}'$ are defined in \eqref{mu'}. Then using bounds \eqref{Kmu-bound} and analogous ones for the dual functions
\begin{align}
&| \mu  (y) | \leq C e^{ \pi\nu_g |y| }, \qquad \; \, |\K (y) |\leq C e^{ - \pi\nu_g|y| }, \qquad \;\, y \in \mathbb{R},\\[5pt]
&| \hat{\mu}  (\g) | \leq C e^{ \pi \nu_{\hat{g}^*} |\g| },  \qquad |\KK (\g) |\leq C e^{ - \pi\nu_{\hat{g}^*} |\g| }, \qquad \g \in \mathbb{R},
\end{align}
where
\beq
\nu_g = \Re  \hat{g}>0, \qquad \nu_{\hat{g}^*} = \Re g^*>0,
\eeq
together with triangle inequalities
\begin{equation}
|x_i - y_j | \geq |y_j| - |x_i|, \qquad |\lambda_i - \gamma_j| \geq |\gamma_j| - |\lambda_i|
\end{equation}
we bound part of the integrand by exponentially decreasing function
\begin{equation}\label{QQ'bound}
\begin{aligned}
\Bigl| \mu' (-\bm{y}_{n - 1}) \, \hat{\mu}'(-\bm{\gamma}_{n - 1}) &\, \K (\bm{x}_{n - 1}, \bm{y}_{n - 1}) \, \KK (\bm{\lambda}_{n - 1}, \bm{\gamma}_{n - 1}) \Bigr| \\[5pt]
&\leq C_q \exp \Bigl( - \pi\nu_g \| \bm{y}_{n - 1} \| -  \pi\nu_{\hat{g}^*} \| \bm{\gamma}_{n - 1} \| \Bigr)
\end{aligned}
\end{equation}
with some $C_q(g, \bo, \bm{x}_{n - 1}, \bm{\lambda}_{n - 1})$.

For the rest part of the integrand (besides surely bounded oscillating exponents) we use estimate \eqref{Psi-HRbound}. If $\Re g\in (0, \max(\omega_1, \omega_2)]$ we use it, as it is written. If $\Re g \in (\max(\omega_1, \omega_2), \omega_1 + \omega_2)$, then $\gg \in (0, \min(\hat{\o}_1,\hat{\o}_2))$. In this case we use an analogous to \eqref{Psi-HRbound} estimate for the dual wave function $\hat{\Psi}_{ \bm{y}_{n - 1} }(\bm{\gamma}_{n - 1})$
\begin{equation}
\Bigl| \mu'  (\bm{y}_{n - 1})\, \hat{\mu}'(\bm{\gamma}_{n - 1})  \, \hat{\Psi}_{ \bm{y}_{n - 1} }(\bm{\gamma}_{n - 1}) \Bigr| \leq P(\bm{\gamma}_{n - 1}; \hat{g}^*|\hat{\bo}),
\end{equation}
since wave functions are equal by induction assumption \eqref{Psi-ind}. In both cases exponent in \eqref{QQ'bound} suppresses the polynomial growth. Thus, the whole integrand is absolutely integrable. This gives the proof of Theorem \ref{thm-duality} for real positive periods.

Since the estimate for multiple integral \rf{int3} presenting the wave function in Proposition \ref{Prop-Qconv} is uniform on compact subsets of the periods, the wave function $\Psi_{\bl_n}(\bx_n)$ depends analitically on the periods $\bo$. The same is true for the function  $\hat{\Psi}_{\bx_n}(\bl_n)$. Hence, their coincidence for real periods implies their coincidence for complex periods. This ends the proof of Theorem \ref{thm-duality}.
\end{proof}

Note that during the proof  we established the third, mixed integral representation of the wave function.
\PsiQQ*

\section*{Acknowledgments}
The work of N. Belousov and S. Derkachov (Sections 1, 3, Appendices A, B) was supported by the Theoretical Physics and Mathematics Advancement Foundation BASIS.
The  work of S. Kharchev (Section 2) was supported by the Russian Science Foundation (Grant No. 23-41-00049). The work of S. Khoroshkin (Section 4) was supported by the International Laboratory of Cluster Geometry of National Research University Higher School of Economics, Russian Federation Government grant, ag. No. 075-15-2021-608 dated 08.06.2021.

\setcounter{equation}{0}

\section*{Appendix}
\appendix
\section{The double sine function} \label{AppA}
The  double sine  function $S_2(z):=S_2(z|\bo)$, see \cite{Ku} and references therein, is a meromorphic function that satisfies two functional relations
\beq\label{trig3}  \frac{S_2(z)}{S_2(z+\o_1)}=2\sin \frac{\pi z}{\o_2},\qquad \frac{S_2(z)}{S_2(z+\o_2)}=2\sin \frac{\pi z}{\o_1}
\eeq
and inversion relation
\beq \label{trig4} S_2(z)S_2(-z)=-4\sin\frac{\pi z}{\o_1}\sin\frac{\pi z}{\o_2},\eeq
or equivalently
\beq S_2(z)S_2(\o_1+\o_2-z)=1. \eeq
The function $S_2(z)$ has poles at the points
\beq \label{S-poles}
z = m \o_1 + k\o_2, \qquad m,k\geq 1
\eeq
and zeros at
\beq\label{S-zeros}
z=-m\o_1-k\o_2,\qquad m,k\geq 0.
\eeq
For $\o_1 / \o_2 \not\in \mathbb{Q}$ all poles and zeros are simple.
In the analytic region $ \Re z \in ( 0, \Re(\omega_1 + \omega_2) )$ we have the following integral representation for the logarithm of $S_2(z)$
\begin{equation}\label{S2-int}
\ln S_2 (z) = \int_0^\infty \frac{dt}{2t} \left( \frac{\sh \left[ (2z - \omega_1 - \omega_2)t \right]}{ \sh (\omega_1 t) \sh (\omega_2 t) } - \frac{ 2z - \omega_1 - \omega_2 }{ \omega_1 \omega_ 2 t } \right).
\end{equation}
It is clear from this representation that the double sine function is homogeneous
\beq\label{S-hom}
S_2( \gamma z | \gamma\o_1, \gamma \o_2 ) = S_2(z|\o_1, \o_2), \qquad \gamma \in (0, \infty)
\eeq
and invariant under permutation of periods
\beq\label{A6}
S_2(z| \o_1, \o_2) = S_2(z | \o_2, \o_1).
\eeq
The double sine function can be expressed through the Barnes double Gamma function $\Gamma_2(z|\bo)$ \cite{B},
\beq
S_2(z|\bo)=\Gamma_2(\o_1+\o_2-z|\bo)\Gamma_2^{-1}(z|\bo),
\eeq
and its properties follow from the corresponding properties of the double Gamma function.
It is also connected to the Ruijsenaars hyperbolic Gamma function $G(z|\bo)$ \cite{R1}
\beq \label{S-G}
G(z|\bo) = S_2\Bigl(\imath z + \frac{\o_1 + \o_2}{2} \,\Big|\, \bo \Bigr)
\eeq
and to the Faddeev quantum dilogarithm $\gamma(z|\bo)$ \cite{F}
\beq 
\gamma(z|\bo) = S_2\Bigl(-\imath z + \frac{\o_1+\o_2}{2}\, \Big|\, \bo\Bigr) \exp \Bigl( \frac{\imath \pi}{2\o_1 \o_2} \Bigl[z^2 + \frac{\o_1^2+\o_2^2}{12} \Bigr]\Bigr).
\eeq
Both $G(z|\bo)$ and $\gamma(z|\bo)$ were investigated independently.

In the paper we deal only with ratios of double sine functions denoted by $\mu(x)$ \eqref{I5} and $K(x)$ \eqref{I6}
\beq\label{B1}
\begin{split}\mu(x)& =S_2(\imath x)S_2^{-1} (\imath x+g),\\[6pt]
	K(x)& =  S_2\left(\imath x+\frac{\o_1+\o_2}{2}+\frac{g}{2}\right)S_2^{-1}\left(\imath x+\frac{\o_1+\o_2}{2}-\frac{g}{2}\right).
\end{split}
\eeq
Now we will give the key asymptotic formulas and bounds for them, which were derived in \cite[Appendices A, B]{BDKK} from the known results for the double Gamma function. In what follows we assume conditions \eqref{I0a}, \eqref{I0b}
\begin{equation}
\Re \o_j > 0, \qquad 0 < \Re g < \Re\o_1 + \Re \o_2, \qquad \Re \hat{g} > 0,
\end{equation}
where we denoted
\begin{equation}
\hat{g} = \frac{g}{\o_1 \o_2}.
\end{equation}
Let $\sigma_i$ be the arguments of the periods~$\o_i$, $|\sigma_i|<\pi/2$. Because of the symmetry \eqref{A6}, we may assume that $\sigma_1 \geq \sigma_2$. Let $D_+$ and $D_-$ be the cones of poles \eqref{S-poles} and zeros \eqref{S-zeros} of~$S_2(z)$:
\beqq D_+=\{ z\colon \sigma_2< \arg z<\sigma_1\},\qquad D_-=\{ z\colon \pi+\sigma_2< \arg z<\pi+\sigma_1\},\qquad D=D_+\cup D_-.\eeqq
Denote by $d(z,D)$ the distance between a point $z$ and the cones $D$. Using Barnes' Stirling formula for the asymptotic of double Gamma function, for the ratio of double sines one obtains
\beq \label{A14} \frac{S_2(z)}{S_2(z+g)} =  e^{\mp\pi\imath \hat{g}\left(z-\frac{g^\ast}{2}\right)}\Bigl(1+\,O\Bigl(d^{-1}(z,D)\Bigr)\Bigr), \eeq
where $z \in \mathbb{C} \setminus D$ and the sign $-$ (or $+$) is taken for $\Re z > 0$ (or $\Re z < 0$), see \cite[eq.(A.18)]{BDKK}.
Then from \eqref{A14} for the functions $\mu(x)$ and $K(x)$ \eqref{B1} with $x \in \R$ we have
\begin{equation}\label{Kmu-asymp}
\mu(x) \sim e^{\pi \hat{g} | x | \pm \imath \frac{\pi \hat{g} g^*}{2} }, \qquad K(x) \sim e^{- \pi \hat{g} |x|}, \qquad x\rightarrow \pm \infty.
\end{equation}
Denote also
\beq \nu_g= \Re  \hat{g}.\eeq
Under the conditions \eqref{I0a}, \eqref{I0b} by using the same Stirling formula we also have bounds
\beq \label{Kmu-bound}
|\mu(x)| \leq C e^{\pi\nu_g |x|}, \qquad |K(x)| \leq C e^{-\pi\nu_g |x|},  \qquad x \in \R
\eeq
where $C$ is a positive constant uniform for compact subsets of parameters $\bo, g$ preserving the mentioned conditions, see \cite[eq.(B.3)]{BDKK}.

Another key result that we need in the paper is the following Fourier transform formula given in \cite[Proposition C.1]{R3}, which we rewrite in terms of the double sine function using connection formula \eqref{S-G}.  This Fourier transform can be already found in \cite{FKV,PT}.
\begin{proposition*}\cite{R3}
For real positive periods $\o_1, \o_2$ we have
\begin{equation}
\begin{aligned}
&\int_{\mathbb{R}} dx \, e^{\frac{2\pi\imath}{\o_1 \o_2} y x} S_2\Bigl(\imath x - \imath \nu + \frac{\o_1 + \o_2}{2} \Bigr) S_2^{-1} \Bigl( \imath x - \imath \rho + \frac{\o_1 + \o_2}{2} \Bigr) \\[6pt]
&= \sqrt{\o_1 \o_2} \, e^{ \frac{\pi\imath}{\o_1 \o_2} y (\nu + \rho) } S_2(\imath \rho - \imath \nu) \, S_2^{-1}\Bigl(\imath y + \frac{\imath(\rho - \nu)}{2} \Bigr) \, S_2^{-1}\Bigl(-\imath y + \frac{\imath(\rho - \nu)}{2} \Bigr),
\end{aligned}
\end{equation}
while the parameters $\nu, \rho, y$ satisfy the conditions
\begin{equation}\label{A19}
-\frac{\o_1 + \o_2}{2} < \Im \rho < \Im \nu < \frac{\o_1 + \o_2}{2}, \qquad |\Im y | < \Im \frac{\nu - \rho}{2}.
\end{equation}
\end{proposition*}
In the special case
\begin{equation}
\nu = \frac{\imath g}{2}, \qquad \rho = -\frac{\imath g}{2}
\end{equation}
taking $y = \o_1 \o_2 \l$ and using homogeneity of the double sine \eqref{S-hom} (with $\gamma = \o_1 \o_2$) we arrive at the Fourier transform formula for the function $K(x)$ \eqref{B1}
\begin{equation}\label{K-fourier}
\int_{\mathbb{R}} dx \; e^{2 \pi \imath  \lambda x }  \K (x) = \sqrt{\omega_1 \omega_2} \, S_2(g) \, \KK (\lambda),
\end{equation}
where $| \Im \l | < \Re \hat{g}/2$ and conditions \eqref{A19} are satisfied due to the inequalities on the coupling constant $g$ \eqref{I0a}, \eqref{I0b}. Here we recall the notations \eqref{I27a}, \eqref{I6c}
\begin{equation}
\hat{K}(\l) = K_{\hat{g}^*}(\l|\hat{\o}), \qquad \hat{g}^* = \frac{g^*}{\o_1\o_2}, \qquad \hat{\bo} = \Bigl( \frac{1}{\o_2}, \frac{1}{\o_1} \Bigr).
\end{equation}
Note that the right hand side of \eqref{K-fourier} is analytic function of $\o_1, \o_2$ in the domain $\Re \o_j > 0$. The integral from the left is also analytic with respect to periods. Indeed, due to the bound \eqref{Kmu-bound} it is absolutely convergent uniformly on compact sets of parameters $\bo, g$ preserving the conditions \eqref{I0a}, \eqref{I0b}. Hence, the formula \eqref{K-fourier} also holds for complex periods under the mentioned conditions.
\cbk

\section{Inequalities with absolute value}
First, we will prove a little lemma that will be crucial for all other estimates.
\begin{lemma}\label{ineq}
For any $\epsilon \in [0, 2]$, $y_1, y_2, y \in \mathbb{R}$ we have
\begin{equation}
|y_1 - y_2| - |y_1 - y| - |y_2 - y| \leq \epsilon \left( |y_1| + |y_2| - |y| \right).
\end{equation}
\end{lemma}
\begin{proof}
Consider two cases. First, assume that $|y_1| + |y_2| \geq |y|$. Then using triangle inequality we obtain
\begin{equation}
|y_1 - y_2| - |y_1 - y| - |y_2 - y| \leq 0 \leq \epsilon \left( |y_1| + |y_2| - |y| \right)
\end{equation}
since $\epsilon \geq 0$.

Second, assume $|y_1| + |y_2| \leq |y|$. Now using $|y_j - y| \geq |y| - |y_j|$ and $|y_1 - y_2| \leq |y_1| + |y_2|$ we arrive at
\begin{equation}
|y_1 - y_2| - |y_1 - y| - |y_2 - y| \leq 2 \left( |y_1| + |y_2| - |y| \right) \leq \epsilon \left( |y_1| + |y_2| - |y| \right)
\end{equation}
because $\epsilon \leq 2$.
\end{proof}

\begin{remark}
Note that Lemma \ref{ineq} holds for any norm and points $y_j$ in the corresponding normed space.
\end{remark}

Denote by $\bm{y}_k$ a vector with $k$ components $\bm{y}_k = \bigl( y^{(k)}_{1}, \dots, y^{(k)}_{k} \bigr)$, $y_j^{(k)}\in \R$. Define a function $S_n$ by a recurrence relation
\begin{equation}\label{Srec2}
\begin{aligned}
S_n ( \bm{y}_1, \dots, \bm{y}_n ) = \sum_{\substack{i, j = 1 \\ i \not= j}}^n \bigl| y^{(n)}_i - y^{(n)}_j \bigr| &- \sum_{i = 1}^{n} \sum_{j = 1}^{n - 1} \, \bigl| y^{(n)}_i - y^{(n - 1)}_j \bigr| \\
&+S_{n - 1} ( \bm{y}_1, \dots, \bm{y}_{n - 1} )
\end{aligned}
\end{equation}
with $S_1 = 0$. Also by $\| \bm{y}_k \|$ denote $L^1$-norm
\begin{equation}
\| \bm{y}_k \| = \sum_{j = 1}^k \bigl| y^{(k)}_j \bigr|.
\end{equation}

\begin{lemma}\label{Slemma}
The inequality
\begin{equation}\label{Sest}
S_n \leq \frac{1}{2}\sum_{\substack{i, j = 1 \\ i \not= j}}^n \bigl| y^{(n)}_i - y^{(n)}_j \bigr| + c_n \epsilon \| \bm{y}_n \|  - \epsilon \sum_{k = 1}^{n - 1} \| \bm{y}_k \|,
\end{equation}
holds for any $\epsilon \in \left[0, \frac{2(n - 1)}{c_n} \right]$, where the number $c_n$ is defined by recurrence relation $c_n = (n - 1) (1 + c_{n - 1})$ with $c_1 = 0$.
\end{lemma}

\begin{proof}
The proof goes by induction. Consider the case $n = 2$. Using Lemma~\ref{ineq} we have
\begin{equation}
\begin{aligned}
S_2 = 2 \bigl| y_1^{(2)} - y_2^{(2)} \bigr| - \bigl| y_1^{(2)} - & y_1^{(1)} \bigr| - \bigl| y_2^{(2)} - y_1^{(1)} \bigr| \\[6pt]
&\leq \bigl| y_1^{(2)} - y_2^{(2)} \bigr| + \epsilon \left( \bigl| y_1^{(2)} \bigr| + \bigl| y_2^{(2)} \bigr| \right) - \epsilon \bigl| y_1^{(1)} \bigr|
\end{aligned}
\end{equation}
for any $\epsilon \in [0, 2]$. Thus, we proved the base case.

Now suppose we proved the statement for $S_{n - 1}$. Both sides of the stated inequality are symmetric with respect to components of $\bm{y}_k$ for all $k = 1, \dots, n$. Therefore, without loss of generality we can assume the ordering
\begin{equation}\label{ord}
y^{(k)}_1 \geq \ldots \geq y^{(k)}_k, \qquad k = 1, \dots, n.
\end{equation}
For the vector $\bm{y}_n$ with ordered components we can write
\begin{equation}\label{ord-comp}
\sum_{\substack{i, j = 1 \\ i \not= j}}^n \bigl| y^{(n)}_i - y^{(n)}_j \bigr| = 2 \sum_{m = 1}^{\lfloor n/2 \rfloor} (n - 2m + 1) \bigl| y_m^{(n)} - y_{n - m + 1}^{(n)} \bigr|.
\end{equation}
Consequently, due to the recurrence relation \eqref{Srec2} and induction assumption
\begin{equation}\label{Sineq}
\begin{aligned}
S_n &\leq 2 \sum_{m = 1}^{\lfloor n/2 \rfloor} (n - 2m + 1) \bigl| y_m^{(n)} - y_{n - m + 1}^{(n)} \bigr| - \sum_{i = 1}^n \sum_{j = 1}^{n - 1} \, \bigl| y^{(n)}_i - y^{(n - 1)}_j \bigr| \\[5pt]
&+ \sum_{m = 1}^{\lfloor (n - 1)/2 \rfloor} (n - 2m) \bigl| y_m^{(n - 1)} - y_{n - m}^{(n - 1)} \bigr|  + c_{n - 1} \epsilon \| \bm{y}_{n - 1} \| - \epsilon \sum_{k = 1}^{n - 2} \| \bm{y}_k \|.
\end{aligned}
\end{equation}
Next step it to regroup and estimate terms from the first three sums. For the term with $m = 1$ from the first sum we use the Lemma \ref{ineq} and for all other terms we simply use triangle inequalities.

Consider term with $m = 1$ from the first sum and terms with $i = 1, n$ from the second double sum and write the estimate
\begin{equation}\label{ineq1}
\begin{aligned}
(n - 1) \bigl| y_1^{(n)} - y_{n}^{(n)} \bigr| - \sum_{j = 1}^{n - 1} & \left( \bigl| y_1^{(n)} - y^{(n - 1)}_j \bigr| + \bigl| y_{n}^{(n)} - y^{(n - 1)}_j \bigr| \right) \\[4pt]
& \leq (n - 1) \, \epsilon' \left( \bigl| y_1^{(n)} \bigr| + \bigl| y_n^{(n)} \bigr| \right) - \epsilon' \, \| \bm{y}_{n - 1} \|,
\end{aligned}
\end{equation}
where we used Lemma \ref{ineq} with parameter $\epsilon' \in [0, 2]$ multiple times. Similarly let us estimate the term with $m > 1$ from the first sum together with the corresponding terms from the second double sum
\begin{equation}\label{ineq2}
\begin{aligned}
(n - 2m + 1) &\bigl| y_m^{(n)} - y_{n - m + 1}^{(n)} \bigr| \\[4pt]
&- \sum_{j = m}^{n - m} \left( \bigl| y_m^{(n)} - y^{(n - 1)}_j \bigr| + \bigl| y_{n - m + 1}^{(n)} - y^{(n - 1)}_j \bigr| \right) \leq 0,
\end{aligned}
\end{equation}
where we used triangle inequality multiple times. Remaining from the second double sum terms can be grouped with terms from the third sum
\begin{equation}\label{ineq3}
\begin{aligned}
(n - 2m) & \bigl| y_m^{(n - 1)} - y_{n - m}^{(n - 1)} \bigr| \\[4pt]
&- \sum_{i = m + 1}^{n - m} \left( \bigl| y_i^{(n)} - y^{(n - 1)}_m \bigr| + \bigl| y_i^{(n)} - y^{(n - 1)}_{n - m} \bigr| \right) \leq 0,
\end{aligned}
\end{equation}
where we again used triangle inequalities.

So, starting from \eqref{Sineq} and using inequalities \eqref{ineq1}, \eqref{ineq2}, \eqref{ineq3} we obtain
\begin{equation}
\begin{aligned}
S_n \leq \sum_{m = 1}^{\lfloor n/2 \rfloor} (n - 2m + 1) \bigl| y_m^{(n)} - y_{n - m + 1}^{(n)} \bigr| + (n - 1) \epsilon' \left( \bigl| y_1^{(n)} \bigr| + \bigl| y_n^{(n)} \bigr| \right) \\
+ (c_{n - 1} \epsilon - \epsilon') \| \bm{y}_{n - 1} \| - \epsilon \sum_{k = 1}^{n - 2} \| \bm{y}_k \|.
\end{aligned}
\end{equation}
Finally, for the first sum we again use equality \eqref{ord-comp}. For the second sum we have
\begin{equation}
\bigl| y_1^{(n)} \bigr| +  \bigl| y_{n}^{(n)} \bigr| \leq \| \bm{y}_{n} \|.
\end{equation}
Choosing $\epsilon' = (1 + c_{n - 1} )\epsilon$ we get bound stated in the lemma. Note that since $\epsilon' \leq 2$ then $\epsilon \leq 2/(1 + c_{n - 1}) = 2(n - 1)/c_n$.
\end{proof}

\setcounter{corollary}{4}
\begin{corollary}
The inequality
\begin{equation}\label{Sest2}
S_n \leq (n - 1 +  \epsilon) \| \bm{y}_n \| - \frac{\epsilon}{c_n} \sum_{k = 1}^{n - 1} \| \bm{y}_k \|
\end{equation}
holds for any $\epsilon \in \left[ 0, 2(n - 1) \right]$.
\end{corollary}

\begin{proof}
Use Lemma \ref{Slemma} rescaling the parameter $\ve \rightarrow \ve/c_n$ and
\begin{equation}
\frac{1}{2} \sum_{\substack{i, j = 1 \\ i \not= j}}^n \bigl| y^{(n)}_i - y^{(n)}_j \bigr| \leq (n - 1) \| \bm{y}_n \|
\end{equation}
to arrive at stated inequality.
\end{proof}

In fact, Lemma \ref{Slemma} gives an estimate for the integrand in the integral representation of the wave function $\Psi_{\bl_n}(\bx_n)$ \eqref{I12} and, as a consequence, the exponential bound for the wave function. Denote, as before \eqref{mu'},
\beq \mu'(\bx_n) = \prod_{\substack{i,j=1 \\ i<j}}^n \mu(x_i - x_j).\eeq
\begin{corollary}
The wave function admits the bound
\beq
\bigl| \mu'(\bx_n) \Psi_{\bl_n}(\bx_n) \bigr| \leq C \exp \pi\nu_g \left( \ve \| \bx_n \| -\frac{2}{\nu_g} \bbx_n \Im\l_n   \right)
\eeq
with some $C(g, \bo)$ for any $\ve \in(0, 2(n - 1)]$, assuming $x_j \in \R$ and
\begin{equation}
| \Im(\l_k - \l_j) |< \frac{\nu_g}{2(n - 1)!e} \, \ve, \qquad k,j=1, \dots,n.
\end{equation}
\end{corollary}
\begin{proof}
The function $\mu'(\bx_n)$ is estimated as
\begin{equation}\label{B43}
|\mu'(\bx_n)| \leq C_1 \exp \pi \nu_g \Biggl( \; \sum_{\substack{i,j=1 \\ i<j}}^n |x_i-x_j| \, \Biggr)
\end{equation}
with some $C_1(g, \bo)$, where we used bound from \eqref{Kmu-bound}.
The integral representation \eqref{I12} in full form
\begin{equation}
\begin{aligned}
\Psi_{\bm{\lambda}_n} (\bm{x}_n) = C_\Psi \int_{\mathbb{R}^{n(n-1)/2} } & d\bm{y}_{n - 1} \dots d\bm{y}_1 \; e^{ {2\pi \imath}{} \left[ \lambda_n \bbx_n + (\lambda_{n - 1} - \lambda_{n}) \bby_{n - 1} \right] } \K (\bm{x}_n, \bm{y}_{n - 1}) \\[5pt]
& \hspace{0.6cm} \times  \prod_{k = 2}^{n - 1} e^{{2\pi \imath}{}(\lambda_{k - 1} - \lambda_{k}) \bby_{k - 1}} \, \mu   ( \bm{y}_{k} ) \, \K  (\bm{y}_{k}, \bm{y}_{k - 1} ),
\end{aligned}
\end{equation}
where $C_\Psi$ contains all constants $d_k$. Denote the integrand by $H$. Assuming
\begin{equation}
|\Im(\l_{k - 1} - \l_{k})| \leq \delta_\Lambda \frac{\nu_g}{2}
\end{equation}
and using bounds \eqref{Kmu-bound} we arrive at
\begin{equation}
\begin{aligned}
|H| \leq C_2 \exp \pi\nu_g \Biggl(  -\frac{2}{\nu_g} \bbx_n \Im\l_n  &+\delta_\Lambda \sum_{k = 1}^{n - 1}\|\by_k\| \\
&+ S_n(\by_1, \dots,\by_{n - 1}, \bx_n) - \sum_{\substack{i,j=1\\i\not= j}}^n|x_i - x_j| \Biggr)
\end{aligned}
\end{equation}
with some $C_2(g,\bo)$. Using Lemma \ref{Slemma} with rescaling $\ve \rightarrow \ve/c_n$ we have
\begin{equation}\label{B47}
|H| \leq C_2 \exp \pi\nu_g \Biggl(  -\frac{2}{\nu_g} \bbx_n \Im\l_n  + \Bigl(\delta_\Lambda - \frac{\ve}{c_n} \Bigr) \sum_{k = 1}^{n - 1}\|\by_k\| + \ve \|\bx_n\| - \frac{1}{2} \sum_{\substack{i,j=1\\i\not= j}}^n|x_i - x_j| \Biggr)
\end{equation}
for any $\ve \in [0, 2(n - 1)]$. For
\begin{equation}
\delta_\Lambda < \frac{\ve}{(n - 1)! e} < \frac{\ve}{c_n},
\end{equation}
see \eqref{kh5}, the bound \eqref{B47} represents integrable function. Combining it with \eqref{B43} we arrive at the bound stated in the corollary.
\end{proof}

For the next inequality define another function that depends on vectors $\bm{y}_k$ and additional vector $\bm{t}_n = (t_1, \dots, t_n)$
\begin{equation}\label{Tdef}
T_n(\bm{y}_1, \dots, \bm{y}_n, \bm{t}_n) = \sum_{\substack{i, j = 1 \\ i \not= j}}^n | t_i - t_j | - \sum_{i, j = 1}^n \bigl| t_i - y_j^{(n)} \bigr| + S_n(\bm{y}_1, \dots, \bm{y}_n).
\end{equation}

\begin{lemma}
The inequality
\begin{equation}\label{Test}
\begin{aligned}
T_n \leq ( n + r ) \| \bm{t}_n \| - \frac{1 - r}{2 n \, c_n} \sum_{k = 1}^{n} \| \bm{y}_k \| - r \, \biggl| \sum_{j = 1}^n \bigl( t_j - y_j^{(n)} \bigr) \biggr|
\end{aligned}
\end{equation}
holds for any $r \in [0, 1]$.
\end{lemma}

\begin{proof}
Both sides of the stated inequality are symmetric with respect to components of the vectors $\bm{y}_k, \bm{t}_n$ (separately). Therefore, without loss of generality we assume ordering
\begin{equation}\label{ord2}
y^{(k)}_1 \geq \ldots \geq y_k^{(k)}, \qquad t_1 \geq \ldots \geq t_n, \qquad k = 1, \dots, n.
\end{equation}
Then, as in the proof of the previous lemma, we can write
\begin{equation}\label{tord}
\sum_{\substack{i, j = 1 \\ i \not= j}}^n | t_i - t_j | = 2 \sum_{m = 1}^{\lfloor n / 2 \rfloor} (n - 2m + 1) | t_m - t_{n - m + 1} |.
\end{equation}
Next, take arbitrary $r \in [0, 1]$ and use triangle inequalities for the second double sum in \eqref{Tdef} to obtain
\begin{equation}\label{tr}
- \sum_{i, j = 1}^n \bigl| t_i - y_j^{(n)} \bigr| \leq - \left(1 - \frac{r}{n} \right) \sum_{i, j = 1}^n \bigl| t_i - y_j^{(n)} \bigr| - r \, \biggl| \sum_{j = 1}^n \bigl( t_j - y_j^{(n)} \bigr) \biggr|.
\end{equation}
Combining two previous formulas \eqref{tord}, \eqref{tr} and Lemma \ref{Slemma} with the parameter \vspace{3pt}
\begin{equation}
\epsilon = \frac{1 - r}{2n \, c_n} \, \in \, \left[ 0, \frac{2(n - 1)}{c_n} \right]
\end{equation}
we arrive at inequality
\begin{equation}\label{Tineq1}
\begin{aligned}
T_n &\leq \sum_{m = 1}^{\lfloor n / 2 \rfloor} (n - 2m + 1) \left( 2 \, | t_m - t_{n - m + 1} | + \bigl| y_m^{(n)} - y^{(n)}_{n - m + 1} \bigr| \right) \\[6pt]
&- \left(1 - \frac{r}{n} \right) \sum_{i, j = 1}^n \bigl| t_i - y_j^{(n)} \bigr| + \frac{1 - r}{2n} \| \bm{y}_n \| - \frac{1 - r}{2n \, c_n} \sum_{k = 1}^{n - 1} \| \bm{y}_k \| - r \, \biggl| \sum_{j = 1}^n \bigl( t_j - y_j^{(n)} \bigr) \biggr|.
\end{aligned}
\end{equation}
Notice that two last sums from the right coincide with the corresponding terms in the stated inequality \eqref{Test}, so it is left to bound the rest ones. Denote them by $R_n$
\begin{equation}\label{Rn}
\begin{aligned}
R_n = \sum_{m = 1}^{\lfloor n / 2 \rfloor} (n - 2m & + 1) \left( 2 \, | t_m - t_{n - m + 1} | + \bigl| y_m^{(n)} - y^{(n)}_{n - m + 1} \bigr| \right) \\[6pt]
&- \left(1 - \frac{r}{n} \right) \sum_{i, j = 1}^n \bigl| t_i - y_j^{(n)} \bigr| + \frac{1 - r}{2n} \| \bm{y}_n \|.
\end{aligned}
\end{equation}
We wish to prove the inequality
\begin{equation}\label{Rineq}
R_n \leq (n + r) \| \bm{t}_n \| - \frac{1 - r}{2n \, c_n} \| \bm{y}_n \|.
\end{equation}

Now fix $m \in \{ 1, \dots, \lfloor n / 2 \rfloor \}$ and from all of the remaining terms $R_n$ consider the following ones
\begin{equation}\label{m}
\begin{aligned}
R_{n, m} &= 2(n - 2m + 1) |t_m - t_{n - m + 1}| + (n - 2m + 1) \bigl| y_m^{(n)} - y^{(n)}_{n - m + 1} \bigr| \\[5pt]
& -\left(1 - \frac{r}{n} \right) \sum_{i = m}^{n - m + 1} \left( \bigl| t_i - y_m^{(n)} \bigr| + \bigl| t_i - y_{n - m + 1}^{(n)} \bigr| \right) \\[5pt]
&- \left(1 - \frac{r}{n} \right) \sum_{j = m + 1}^{n - m} \left( \bigl| t_m - y_j^{(n)} \bigr| + \bigl| t_{n - m + 1} - y_j^{(n)} \bigr| \right) + \frac{1 - r}{2n} \left( \bigl| y_m^{(n)} \bigr| + \bigl| y_{n - m + 1}^{(n)} \bigr| \right).
\end{aligned}
\end{equation}
Then for even $n$ we have
\begin{equation}\label{Revenn}
R_n = \sum_{m = 1}^{\lfloor n / 2 \rfloor} R_{n, m}
\end{equation}
and for odd $n$
\begin{equation}\label{Roddn}
R_n = \sum_{m = 1}^{\lfloor n / 2 \rfloor} R_{n, m} - \left( 1 - \frac{r}{n} \right) \bigl| t_{\frac{n + 1}{2}} - y_{\frac{n + 1}{2}}^{(n)} \bigr| + \frac{1 - r}{2 n} \, \bigl| y_{\frac{n + 1}{2}}^{(n)} \bigr|,
\end{equation}
since in this case two variables $t_{(n + 1)/2}$, $y_{(n + 1)/2}^{(n)}$ don't enter the first sum in \eqref{Rn}.

Let us prove the following bound
\begin{equation}\label{Rnm}
R_{n, m} \leq ( n + r ) \left( |t_m| + |t_{n - m + 1}| \right) - \frac{1 - r}{2n \, c_n} \left( \bigl| y_m^{(n)} \bigr| + \bigl| y_{n - m + 1}^{(n)} \bigr| \right).
\end{equation}
Note that for even $n$ the inequality we wish to prove \eqref{Rineq} directly follows from it due to \eqref{Revenn}. For odd $n$ we have two more terms in \eqref{Roddn}, but they can be easily bounded using triangle inequality
\begin{equation}
\begin{aligned}
- \left( 1 - \frac{r}{n} \right) \bigl| t_{\frac{n + 1}{2}} - y_{\frac{n + 1}{2}}^{(n)} \bigr| + \frac{1 - r}{2 n}  \, \bigl| y_{\frac{n + 1}{2}}^{(n)} \bigr|  &\leq \left( 1 - \frac{r}{n} \right) \bigl| t_{\frac{n + 1}{2}} \bigr| + \frac{1 + r - 2n}{2n} \, \bigl| y_{\frac{n + 1}{2}}^{(n)} \bigr| \\[10pt]
& \leq (n + r ) \bigl| t_{\frac{n + 1}{2}} \bigr| - \frac{1 - r}{2 n \, c_n}  \, \bigl| y_{\frac{n + 1}{2}}^{(n)} \bigr|.
\end{aligned}
\end{equation}

The proof of \eqref{Rnm} requires several steps. First, using triangle inequalities we can obtain
\begin{equation}\label{Rnm1}
\begin{aligned}
(n - 2m + 1) & \bigl| y_m^{(n)} - y^{(n)}_{n - m + 1} \bigr| \\[5pt]
&- \frac{n - 2m + 1}{n - 2m + 2} \sum_{i = m}^{n - m + 1} \left( \bigl| t_i - y_m^{(n)} \bigr| + \bigl| t_i - y_{n - m + 1}^{(n)} \bigr| \right) \leq 0.
\end{aligned}
\end{equation}
For the rest of the first sum in \eqref{m} we use obvious inequality
\begin{equation}\label{Rnm2}
\begin{aligned}
-\biggl(  \frac{1}{n - 2m + 2} - \frac{r}{n} \biggr)  & \sum_{i = m}^{n - m + 1} \left( \bigl| t_i - y_m^{(n)} \bigr| + \bigl| t_i - y_{n - m + 1}^{(n)} \bigr| \right)\\[7pt]
&\hspace{-0.5cm} \leq -\frac{1 - r}{n} \, \Bigl( \bigl| t_m - y_m^{(n)} \bigr| + \bigl| t_m - y_{n - m + 1}^{(n)} \bigr|  \\[5pt]
&\hspace{2cm} + \bigl| t_{n - m + 1} - y_m^{(n)} \bigr| + \bigl| t_{n - m + 1} - y_{n - m + 1}^{(n)} \bigr| \Bigr).
\end{aligned}
\end{equation}
Then, we take the right-hand side of this formula and estimate it together with a part of $|t_m - t_{n - m + 1}|$ term using Lemma \ref{ineq} with the parameter $\epsilon$ = 1
\begin{equation}\label{Rnm3}
\begin{aligned}
\frac{2(1 - r)}{n} \, |t_m & - t_{n - m + 1}| - \frac{1 - r}{n} \, \Bigl( \bigl| t_m - y_m^{(n)} \bigr|  \\[7pt]
& + \bigl| t_m - y_{n - m + 1}^{(n)} \bigr|  + \bigl| t_{n - m + 1} - y_m^{(n)} \bigr| + \bigl| t_{n - m + 1} - y_{n - m + 1}^{(n)} \bigr| \Bigr) \\[7pt]
&\hspace{1.5cm} \leq \frac{1 - r}{n} \left( 2 |t_m| + 2 |t_{n - m + 1}| - \bigl| y_{m}^{(n)} \bigr| - \bigl| y_{n - m + 1}^{(n)} \bigr| \right).
\end{aligned}
\end{equation}
To bound the second sum in \eqref{m} we again use triangle inequalities
\begin{equation}\label{Rnm4}
\begin{aligned}
(n - 2m) \biggl(1 & - \frac{r}{n} \biggr) |t_m - t_{n - m + 1}| \\[5pt]
&- \left(1 - \frac{r}{n} \right) \sum_{j = m + 1}^{n - m} \left( \bigl| t_m - y_j^{(n)} \bigr| + \bigl| t_{n - m + 1} - y_j^{(n)} \bigr| \right) \leq 0.
\end{aligned}
\end{equation}
Collecting results of the four inequalities \eqref{Rnm1}, \eqref{Rnm2}, \eqref{Rnm3}, \eqref{Rnm4} we obtain the bound
\begin{equation}
\begin{aligned}
R_{n, m} \leq \biggl( n + r & + 2 - 2m - \frac{2 m r}{n} - \frac{2(1 - r)}{n} \biggr) | t_m - t_{n - m + 1} | \\[7pt]
&+ \frac{2 (1 - r)}{n} \left( |t_m| + |t_{n - m + 1}| \right) - \frac{1 - r}{2 n} \left( \bigl| y_{m}^{(n)} \bigr| + \bigl| y_{n - m + 1}^{(n)} \bigr| \right).
\end{aligned}
\end{equation}
Finally, to arrive at the inequality \eqref{Rnm} we use $|t_m - t_{n - m + 1}| \leq |t_m| + |t_{n - m + 1}|$ together with
\begin{equation}
2 - 2m - \frac{2 m r}{n} \leq 0, \hspace{1.5cm} \frac{1 - r}{2 n} \geq \frac{1 - r}{2 n \, c_n} \, .
\end{equation}
Thus, we proved the key bound \eqref{Rnm} and, consequently, the lemma.
\end{proof}

\end{document}